\documentclass[sigconf]{acmart}

\copyrightyear{2021} 
\acmYear{2021} 
\setcopyright{acmcopyright}\acmConference[CCS '21]{Proceedings of the 2021 ACM SIGSAC Conference on Computer and Communications Security}{November 15--19, 2021}{Virtual Event, Republic of Korea}
\acmBooktitle{Proceedings of the 2021 ACM SIGSAC Conference on Computer and Communications Security (CCS '21), November 15--19, 2021, Virtual Event, Republic of Korea}
\acmPrice{15.00}
\acmDOI{10.1145/3460120.3484566}
\acmISBN{978-1-4503-8454-4/21/11}

\usepackage{amsfonts}
\usepackage{appendix}
\usepackage{booktabs} 
\usepackage{tabularx}
\usepackage{makecell}
\usepackage{verbatim}
\usepackage{amsgen,amsmath,amstext,amsbsy,amsopn}
\usepackage{cancel}
\usepackage{xspace}
\usepackage{xcolor}
\usepackage{graphicx}
\usepackage{algorithm}
\usepackage{enumitem}
\usepackage[noend]{algpseudocode}
\usepackage{algorithmicx}
\usepackage{hyperref}
\usepackage{tabularx}
\usepackage{multicol}
\usepackage{multirow}
\usepackage{subcaption}
\usepackage{caption}
\algnewcommand{\LeftComment}[1]{\Statex {\color{teal}\textbf{\(\triangleright\) #1}}} 
\algdef{SE}[AS]{As}{EndAs}{\textbf{as}\ }[1]{}%
\algdef{SE}[DOWHILE]{Do}{Until}{\textbf{loop forever}}[1]{\textbf{until}\ #1}%
\algdef{SxnE}[FOR]{ForInitialize}{EndForInitialize}[1]{\textbf{for}\ #1\ \textbf{initialize}:}
\algdef{SxnE}[IF]{Upon}{EndUpon}[1]{\textbf{upon receiving}\ #1\ \algorithmicdo}
\algdef{SxnE}[IF]{Init}{EndInit}{\textbf{Local variables initialization:}}
\algdef{SxnE}[IF]{As}{EndAs}[1]{\textbf{as}\ #1\ }

\newcommand{\nv}{\textsc{NewView}\xspace}
\newcommand{\fs}{\textsc{Forensic Storage}\xspace}
\newcommand{\detector}{\textsc{Detector}\xspace}

\newcommand{\vc}{\textsc{ViewChange}\xspace}
\newcommand{\vdashc}{View-Change\xspace}
\newcommand{\pp}{\text{Proposal-Promotion}\xspace}

\newcommand{\vone}{\textsc{Prepare}\xspace}
\newcommand{\vtwo}{\textsc{Vote2}\xspace}
\newcommand{\reply}{\textsc{Reply}\xspace}

\newcommand{\vthree}{\textsc{Vote3}\xspace}
\newcommand{\vfour}{\textsc{Vote4}\xspace}

\newcommand{\blame}{\textsc{Blame}\xspace}
\newcommand{\commit}{\textsc{Commit}\xspace}
\newcommand{\precommit}{\textsc{Precommit}\xspace}
\newcommand{\eqc}{e_{qc}\xspace}
\newcommand{\ppr}{\textsc{Pre-prepare}\xspace}

\newcommand{\prepare}{\textsc{Prepare}\xspace}
\newcommand{\qctwo}{\text{\it prepareQC}\xspace}
\newcommand{\qcthree}{\text{\it commitQC}\xspace}
\newcommand{\qcprecom}{\text{\it precommitQC}\xspace}
\newcommand{\qcfour}{\text{\it electQC}\xspace}

\newcommand{\hqc}{\text{\it highQC}\xspace}

\newcommand{\pflock}{\text{\it ledger}\xspace}

\newcommand{\info}{\text{\it Info}\xspace}
\newcommand{\hsa}{HotStuff-view\xspace}
\newcommand{\hsb}{HotStuff-hash\xspace}
\newcommand{\hsc}{HotStuff-null\xspace}

\newcommand{\conflictwithin}{\textsc{Request-Proof-within-View}\xspace}
\newcommand{\conflictacross}{\textsc{Request-Proof}\xspace}
\newcommand{\requestproofleader}{\textsc{Request-Proof-of-Leader}\xspace}

\newcommand{\proofacross}{\textsc{Proof-across-View}\xspace}

\newcommand{\thressign}{\ensuremath{aggregate\text{-}sign}}

\newcommand{\skipmsg}{\textsc{Skip}\xspace}
\newcommand{\skipshare}{\textsc{SkipShare}\xspace}
\newcommand{\coinshare}{\textsc{CoinShare}\xspace}
\newcommand{\done}{\textsc{Done}\xspace}
\newcommand{\threscoin}{\textit{threshold-coin}\xspace}
\newcommand{\thressignature}{\textit{threshold-signature}\xspace}
\newcommand{\coin}{coin}
\newcommand{\coinmsg}{\textsc{Proof-of-Leader}\xspace}

\newcommand{\lr}[1]{\langle #1 \rangle}
\newcommand{\signed}[1]{\langle #1 \rangle}

\newcommand{\gradedconsensus}{{\it Graded Consensus}\xspace}
\newcommand{\bba}{\ensuremath{BBA^*}\xspace}
\newcommand{\coinfixtozero}{Coin-Fixed-To-0\xspace}
\newcommand{\coinfixtoone}{Coin-Fixed-To-1\xspace}
\newcommand{\coingenuinelyflipped}{Coin-Genuinely-Flipped\xspace}

\newcommand{\kartik}[1]{{\color{magenta}
\footnotesize[Kartik: #1] }}

\newcommand{\gerui}[1]{{\color{red}
\footnotesize[Gerui: #1] }}

\newcommand{\myparagraph}[1]{\vspace{5pt}\noindent\textbf{#1}}
\newcommand{\ignore}[1]{}
\newcommand{\kartikold}[1]{}
\AtBeginDocument{%
  \providecommand\BibTeX{{%
    \normalfont B\kern-0.5em{\scshape i\kern-0.25em b}\kern-0.8em\TeX}}}

\begin{document}
\fancyhead{}
\title{BFT Protocol Forensics}


\author{Peiyao Sheng}
\authornote{The first two authors contributed equally to this work.}
\email{psheng2@illinois.edu}
\affiliation{%
  \institution{University of Illinois at Urbana Champaign, USA}
  \country{}
}

\author{Gerui Wang}
\authornotemark[1]
\email{geruiw2@illinois.edu}
\affiliation{%
  \institution{University of Illinois at Urbana Champaign, USA}
  \country{}
}

\author{Kartik Nayak}
\email{kartik@cs.duke.edu}
\affiliation{%
  \institution{Duke University, USA}
  \country{}
  }

\author{Sreeram Kannan}
\email{ksreeram@ece.uw.edu}
\affiliation{%
  \institution{University of Washington, USA}
  \country{}
}

\author{Pramod Viswanath}
\email{pramodv@illinois.edu}
\affiliation{%
 \institution{University of Illinois at Urbana Champaign, USA}
 \country{}
 }

\begin{abstract}
Byzantine fault-tolerant (BFT) protocols allow a group of replicas to come to consensus even when some of the replicas are Byzantine faulty. There exist multiple BFT protocols to securely tolerate an optimal number of faults $t$ under different network settings. However, if the number of faults $f$ exceeds $t$ then security could be violated. In this paper we  mathematically formalize the study of  {\em forensic support} of BFT protocols: we aim to identify (with cryptographic integrity) as many of the malicious replicas as possible and in as distributed manner as possible. Our main result is that forensic support of BFT protocols depends heavily on minor implementation details that do not affect the protocol's security or complexity. Focusing on popular BFT protocols (PBFT, HotStuff, Algorand) we exactly characterize their forensic support, showing  that there exist minor variants of each protocol for which the forensic supports vary widely. 
We show strong forensic support capability of LibraBFT, the consensus protocol of Diem cryptocurrency; our lightweight forensic module implemented on a Diem client is open-sourced  \cite{diemopensource} and is under active consideration for deployment in Diem.  Finally, we show that all secure BFT protocols designed for $2t+1$ replicas communicating over a synchronous network forensic support is inherently  nonexistent; this impossibility result holds for all BFT protocols and even if one has access to  the states of all replicas (including Byzantine ones). 
\end{abstract}



\maketitle
\section{Introduction}

Byzantine Fault Tolerant (BFT) protocols, guaranteeing distributed consensus among parties that follow the protocol, are a corner stone of distributed system theory. In the relatively recent context of blockchains, BFT protocols have received renewed attention; new and efficient (state machine replication (SMR)) BFT protocols specifically designed for  blockchains  have been constructed (e.g., Algorand \cite{gilad2017algorand}, HotStuff \cite{yin2019hotstuff}, Streamlet \cite{chan2020streamlet}). The core theoretical security guarantee is that as long as a certain fraction of nodes are ``honest'', i.e., they follow the protocol, then these nodes achieve consensus with respect to (a time evolving) state machine regardless of the Byzantine actions of the remaining malicious nodes. When the malicious nodes are sufficiently numerous, e.g., strictly more than a $1/3$ fraction of nodes in a partially synchronous network, they can ``break security'', i.e., band together to create different views at the honest participants. 

In this paper, we are focused on ``the day after'' \cite{thedayafter}: events {\em after} malicious replicas have successfully mounted a security breach. Specifically, we focus on identifying which of the participating replicas acted maliciously; we refer to this action as ``forensics''. Successful BFT protocol forensics meets two goals:
\begin{itemize}
    \item identify as many of the nodes that acted maliciously as possible with an irrefutable cryptographic proof of culpability; 
    \item identification is conducted as distributedly as possible, e.g., by the individual nodes themselves, with no/limited communication between each other after the security breach.
\end{itemize}

\noindent {\bf Main contributions.} Our central finding is that the forensic possibilities crucially depend on minor implementation details of BFT protocols; the details themselves do not affect protocol  security or performance (latency and communication complexity); we demonstrate our findings  in the context of several popular BFT protocols for  Byzantine Agreement (BA). We present our  findings in the context of a mathematical  and systematic formulation of the  ``forensic support'' of BFT protocols. The forensic support of any BFT protocol is parameterized as a triplet  $(m,k,d)$, that represents the aforementioned  goals of forensic  analysis. The triplet $(m,k,d)$ along with traditional BFT protocol parameters of $(n,t,f)$ is summarized in Table~\ref{tab:notations}. We emphasize that each of the protocol variants is safe and live when $f \le t$ (here $n= 3t+1$  for all the protocols considered) but their forensic supports are quite different. 

\begin{table}[t]
    \centering
    \begin{tabularx}{\columnwidth}{cX} \toprule
        Symbol & \makecell{Interpretation}\\
        \toprule
        $n$ & total number of replicas\\
        \midrule
        $t$ & maximum number of faults for obtaining agreement and termination \\
        \midrule
        $f$ & actual number of faults\\
        \midrule
        $m$ & maximum number of Byzantine replicas under which forensic support can be provided\\
        \midrule
        $k$ & the number of different honest replicas'  transcripts needed  to guarantee a proof of culpability \\
        \midrule
       $d$ & the number of Byzantine replicas that can be held culpable in case of an agreement violation\\\bottomrule
    \end{tabularx}
     \caption{Summary of notations.}
    \label{tab:notations}
\end{table}

\begin{table}[]
\centering
\begin{tabular}{c|c|ccc}
\toprule
\multirow{2}{*}{\textbf{Protocols}} & \multirow{2}{*}{\textbf{\begin{tabular}[c]{@{}c@{}}Forensic\\ Support\end{tabular}}} & \multicolumn{3}{c}{\textbf{Parameters}} \\ \cline{3-5} 
                                    &                                                                                      & $m$       & $k$              & $d$       \\ \toprule
PBFT-PK                             &                                                                                     &       &               &      \\ 
\hsa                                 & Strong                                                                                    & $2t$      & $1$              & $t+1$     \\ 
VABA                                &                                                                                     &       &               &      \\ \midrule
\hsb                                 & Medium                                                                                    & $2t$      & $t+1$     & $t+1$     \\ \midrule
PBFT-MAC                           &                                                                                     &         &              &    $0$     \\

\hsc                                 & None                                                                                    & $t+1$         &          $2t$      &  $ 1$        \\

Algorand                            &                                                                                     &         &               &  $ 0$       \\ \bottomrule
\end{tabular}
\caption{Summary of results; the forensic support values of $d$ are the largest possible and $n=3t+1$ here.}
\label{tab:results}
\end{table}

Security attacks on BFT protocols can be far more subtle than simply double voting or overt equivocation (when the culpable replicas are readily identified), and the forensic analysis is correspondingly subtle. This paper brings to fore a new dimension for secure BFT protocol design  beyond performance (low latency and low communication complexity): forensic support capabilities. 

 We analyze the forensic support of classical and state-of-the-art BFT protocols for Byzantine Agreement. Our main findings, all in the settings of safety violation (agreement violation), are the following, summarized in Table~\ref{tab:results}. 
 \begin{itemize}[leftmargin=*]
     \item {\bf Parameter $\mathbf{d}$.} We show that 
     the number of culpable replicas $d$ that can be identified is either 0 
     or as large as $t+1$. In other words,  
     if at least one replica can be identified then we can also identify the largest possible number,  $t+1$ replicas; the only exception is for \hsc,  where in a successful safety attack, we can identity the culpability of one malicious replica.

     \item {\bf Parameter $\mathbf{m}$.} We show that the maximum number of Byzantine replicas $m$ allowed for nontrivial forensic support (i.e., $d >0$) cannot be more than $2t$.  
     Furthermore, any forensic support  feasible with $m$ is also feasible with $m$ being its largest value $2t$, i.e., if $(m,k,d)$ forensic support is feasible, then $(2t, k,  d)$ is also feasible.
     \item {\bf Parameter $\mathbf{k}$.} Clearly at least one replica's transcript is needed for forensic analysis, so $k=1$ is the least possible value. This suffices for several of the BFT protocol variants. However for \hsb, $k$ needs to be at least  
     $t+1$ for any nontrivial forensic analysis. 

     \item \textbf{Strong forensic support.} The first three items above imply that the strongest possible forensic support is $(2t, 1, t+1)$. Further, the BFT protocols in Table~\ref{tab:results} that achieve any nontrivial forensic support automatically achieve the strongest possible forensics (the only exception is \hsb, for which the forensic support we identified is the best possible).
     
     \item {\bf Impossibility.} 
     For certain variants of BFT protocols (PBFT-MAC, \hsc, and Algorand), even with transcripts from all honest replicas, non-trivial forensics is simply not possible, i.e., $d=0$  even if $m$ is set to its smallest and $k$  set to its largest possible values ($t+1$ and $2t$ respectively);  again, \hsc allows the culpability of a single  malicious replica. 
     
     \item {\bf Practical impact.} Forensic support is of immediate interest to practical blockchain systems; we conduct a forensic support analysis of LibraBFT, the consensus protocol in the new cryptocurrency {\sf Diem}, and show  in-built strong forensic support. We have implemented the corresponding forensic analysis algorithm inside of a {\sf Diem} client and built an associated forensics dashboard; our reference implementation is available open-source \cite{diemopensource} and is under active consideration for deployment in {\sf Diem}.
     
      \item \textbf{BFT with $n=2t+1$.} For a secure BFT protocol operating on a synchronous network, the ideal setting is $n = 2t+1$. For {\em every} such protocol we show that at most one culpable replica can be identified (i.e., $d$ is at most 1) even if we have access to the state of all honest nodes, i.e, $k= t$. 
     
     \end{itemize}

   \myparagraph{Outline.} We describe our results in the context of related work in \S\ref{sec:related}; to the best of our knowledge, this is the first paper to systematically study BFT protocol forensics. The formal forensic support problem statement and security model is in \S\ref{sec:model}. The forensic support of PBFT, HotStuff and Algorand (and variants) are explored in \S\ref{sec:pbft}, \S\ref{sec:basic-hotstuff}, \S\ref{sec:algorand}, respectively. 
   Our forensic study of LibraBFT and implementation of the corresponding forensic protocol is the focus of \S\ref{sec:impl}. In appendix \S\ref{sec:vaba}, we present the forensic support analysis for VABA, a state-of-the-art  efficient BFT for   asynchronous network conditions.  The impossibility of forensic support for all BFT protocols operating in the classical $n=2t+1$ synchronous network setting is shown in \S\ref{sec:lb}. Our choice of the 5  protocols studied here (PBFT, HotStuff, Algorand, LibraBFT, VABA) is made with the goal of  covering a variety of settings: (a) partially synchronous vs   asynchronous; (b) authenticated vs non-authenticated; (c) player replaceable vs irreplaceable;  (d) chained version vs single-shot version; (e) variants that communicate differing amounts of auxiliary information. Stitching the results across the 5 different protocols into a coherent theory of forensic support of abstract BFT protocols is an exciting direction of future work; this is discussed in \S\ref{sec:conclusion}.

\section{Related Work}
\label{sec:related}
\myparagraph{BFT protocols.} 
PBFT~\cite{castro1999practical,castro2002practical} is the first practical BFT SMR protocol in the partially synchronous setting, with quadratic communication complexity of view change. 
HotStuff~\cite{yin2019hotstuff} is the first partially synchronous SMR protocol that enjoys both a linear communication of view change and optimistic responsiveness. Streamlet~\cite{chan2020streamlet} is another SMR protocol known for its simplicity and textbook construction. Asynchronous Byzantine agreement  is solved by \cite{abraham2019asymptotically} with asymptotically optimal communication complexity and round number. Synchronous protocols such as \cite{abraham2019sync,abraham2020optimal} aim at optimal latency.  
Algorand~\cite{gilad2017algorand,chen2019algorand} designs a committee self-selection mechanism, and the Byzantine agreement protocol run by the committee decides the output for all replicas.

\myparagraph{Beyond one-third faults.} The seminal work of   \cite{dwork1988consensus} shows that it is impossible to solve BA when the adversary corrupts one-third replicas for partially  synchronous communication (the same bound holds for SMR). In BFT2F~\cite{li2007beyond}, a weaker notion of safety is defined, and a protocol is proposed such that when the adversary corrupts more than one-third replicas, the weaker notion of safety remains secure whereas the original safety might be violated. However, the weaker notion of safety does not protect the system against common attacks, e.g., double-spending attack in distributed payment systems. Flexible BFT~\cite{malkhi2019flexible} considers the case where clients have different beliefs in the number of faulty replicas and can act to confirm accordingly. Its protocol works when the sum of Byzantine faults and alive-but-corrupt faults, a newly defined type of faults, are beyond one-third. Two recent works~\cite{xiang2021strengthened,kane2021highway} propose BFT SMR protocols that can tolerate more than one-third Byzantine faults after some specific optimistic period. The goal of these works is to mask the effects of faults, even when they are beyond one-third, quite  different from the goals of forensic analysis. 


\myparagraph{Distributed system forensics.} Accountability has been discussed in seminal works~\cite{haeberlen2007peerreview,haeberlen2009fault} for distributed systems in general. In the lens of BFT consensus protocols,  accountability is defined as the ability for a replica to prove the culpability of a certain number of Byzantine replicas in \cite{civit2019polygraph}. Polygraph, a new BFT protocol with high communication complexity $O(n^4)$, is shown to attain this property in \cite{civit2019polygraph}. Reference  \cite{ranchal2020blockchain} extends Polygraph to an SMR protocol, and devises a finality layer to ``merge'' the disagreement. 
Finality and accountability are also discussed in other recent works; examples include Casper~\cite{casper}, GRANDPA~\cite{stewart2020grandpa}, and Ebb-and-flow~\cite{neu2020ebb}. Casper and GRANDPA identify accountability as a central problem and design their consensus protocols around this goal. Ebb-and-flow \cite{neu2020ebb} observes that accountability is immediate in BFT protocols for  safety violations via equivocating votes; however, as pointed out in  \cite{civit2019polygraph}, safety violations can happen during the view change process and this is the step where accountability is far more subtle.     


Reference \cite{civit2019polygraph} argues that PBFT is not accountable and cannot be modified to be accountable without significant change/cost. We point out that the definition of accountability in \cite{civit2019polygraph} is rather narrow:   two replicas with differing views have  to themselves be able to identify culpability of malicious replicas.  On the other hand, in forensic support, we study the number of honest replicas (not necessarily restricted to the specific two replicas which have identified a security breach) that can identify the culpable malicious replicas. Thus the definition of forensic support   is more flexible than the one on accountability. Moreover, our work shows that we can achieve forensic support for protocols such as PBFT and HotStuff without incurring additional communication complexity other than (i) sending a proof of culpability to the client in case of a safety violation, and (ii) the need to use aggregate signatures instead of threshold signatures. 

\section{Problem Statement and Model}
\label{sec:model}
The goal of state machine replication (SMR) is to build a replicated service that takes requests from clients and provides the clients with the interface of a single non-faulty server, i.e., each client receives the same totally ordered sequence of values. To achieve this, the replicated service uses multiple servers, also called replicas, some of which may be Byzantine, where a faulty replica can behave arbitrarily. A secure  state machine replication protocol  satisfies two guarantees.
    %
    \textbf{Safety}: Any two honest replicas cannot output different sequences of values.
 \textbf{Liveness}: A value sent by a client will eventually be output by honest replicas. 

SMR setting also has external validity, i.e., replicas only output non-duplicated values sent by clients. These values are eventually learned by the clients. Depending on the context, a replica may be interested in learning about outputs too. Hence, whenever we refer to a client for learning purposes, it can be an external entity or a server replica. 
A table of notations is in Table~\ref{tab:notations}.  

In this paper, for simplicity, we focus on the setting of  outputting a {\em single value} instead of a sequence of values. The safety and liveness properties of SMR can then be expressed using the following definition:
\begin{definition}[Validated Byzantine Agreement]
\label{defn:vba}
A protocol solves validated Byzantine agreement among $n$ replicas tolerating a maximum of $t$ faults, if it satisfies the following properties:
\begin{enumerate}[leftmargin=*, topsep=0pt]
    \item[] (Agreement) Any two honest replicas output values $v$ and $v'$, then $v = v'$.
    \item[] (Validity)  If an honest replica outputs $v$, $v$ is an externally valid value, i.e., $v$ is signed by a client. 
    \item[] (Termination) All honest replicas start with externally valid values, and all messages sent among them have been delivered, then honest replicas will output a value.
\end{enumerate}
\end{definition}

\myparagraph{Forensic support.} 
Traditionally, consensus protocols provide guarantees only when $f \leq t$. When $f > t$, there can be a safety or liveness violation; this is the setting of study throughout this paper. Our goal is to provide \emph{forensic support} whenever there is a safety violation (or agreement violation) and the number of Byzantine replicas in the system are not too high. In particular, if the actual number of Byzantine faults are bounded by $m$ (for some $m > t$) and there is a safety violation, we can detect $d$ Byzantine replicas using a \emph{forensic protocol}. The protocol takes as input, the transcripts of honest parties, and outputs and \emph{irrefutable} proof of $d$ culprits. With the irrefutable proof, any party (not necessarily in the BFT system) can be convinced of the culprits' identities \emph{without} any assumption on the number of honest replicas. However, if even with transcripts from all honest replicas, no forensic protocol can output such a proof, the consensus protocol has no forensic support (denoted as ``None'' in Table~\ref{tab:results}). Note that when we say a protocol has no forensic support, we are referring to an impossibility w.r.t.\ providing irrefutable proof for $d$ culprits (more precisely, in the context of Definition~\ref{defn:forensic-support}). In a general sense, there are other ways to provide forensics related to hardware, software, and network except for non-repudiation in the protocol.


To provide forensic support, we consider a setting where a client observes the existence of outputs for two conflicting (unequal)  
values.\footnote{We assume all the (honest) replica outputs are eventually learned by the client. In practice, the client may monitor the outputs by periodically communicating with all replicas.} By running a forensic protocol, the client sends (possibly a subset of) these conflicting outputs to all replicas and waits for their replies. Some of these replicas may be ``witnesses'' and may have (partial) information required to construct the irrefutable proof. After receiving responses from the replicas, the client constructs the proof. We denote by $k$ the total number of transcripts from different honest replicas that are stored by the client to construct the proof.  
\begin{definition}
\label{defn:forensic-support}
$(\mathbf{m},\mathbf{k},\mathbf{d})$-{\bf Forensic Support}.  If $t < f \leq m$ and two honest replicas output conflicting values, then using the transcripts of all messages received from $k$ honest replicas during the protocol, a client can provide an irrefutable proof of culpability of at least $d$ Byzantine replicas.
\end{definition}



\myparagraph{Other assumptions.}
We consider forensic support for multiple protocols each with their own network assumptions. For PBFT and HotStuff, we assume a partially synchronous network~\cite{dwork1988consensus}. For VABA~\cite{abraham2019asymptotically} and Algorand~\cite{chen2019algorand}, we suppose  asynchronous and synchronous networks respectively.

We assume all messages are digitally signed except for one variant of PBFT (\S\ref{sec:pbft-mac}) that sometimes relies on the use of Message Authenticated Codes (MACs). Some protocols, e.g., HotStuff, VABA, use threshold signatures. For forensic purposes, we assume  multi-signatures~\cite{bls} instead (possibly worsening the communication complexity in the process). Whenever the number of signatures exceeds a threshold, the resulting aggregate signature is denoted by a pair $\sigma = (\sigma_{agg}, \epsilon)$, where $\epsilon$ 
is a bitmap indicating whose signatures are included in $\sigma_{agg}$. 
We define the intersection of two aggregated messages to be the set of replicas who sign both messages, i.e., $\sigma \cap \sigma' := \{i|\sigma.\epsilon[i]\land \sigma'.\epsilon[i] = 1\}$. An aggregate signature serves as a quorum certificate (QC) in our protocols, and we will use the two terms interchangeably. We assume a collision resistant cryptographic hash function. 
\section{Forensic Support for PBFT}\label{sec:pbft}

PBFT is a classical partially synchronous consensus protocol that provides an optimal resilience of $t$ Byzantine faults out of $n = 3t+1$ replicas. However, when the actual number of faults $f$ exceeds $t$, it does not provide any safety or liveness. In this section, we  show that when $f > t$ and in case of a safety violation,  the variant of the PBFT protocol (referred to as PBFT-PK) where all messages sent by parties are signed, has the \emph{strongest forensic support}. Further, we  show that for an alternative variant where parties sometimes only use MACs (referred to as PBFT-MAC), forensic support is impossible.

\subsection{Overview} 
\label{subsec:pbft-overview}
We start with an overview  focusing on  a single-shot version of PBFT, i.e., a protocol for consensus on a single value. The protocol described here uses digital signatures to authenticate all messages and routes messages through leaders as shown in \cite{ramasamy2005parsimonious}; however we note that our arguments  for PBFT-PK also apply to the original protocol in \cite{castro1999practical}. 

The protocol proceeds in a sequence of consecutive views denoted as view number $e = 1, 2, \cdots$. A higher view is a view with a larger view number. Each view has a unique leader. Each view of PBFT progresses as follows:

\begin{itemize}[leftmargin=*, topsep=0pt]
    \item[-] \textbf{\ppr.} The leader proposes a \nv message containing a proposal $v$ and a status certificate $M$ (explained later) to all replicas.
    \item[-] \textbf{\vone.} On receiving the first \nv message containing a valid value $v$ in a view $e$, a replica sends \vone for $v$ if it is \emph{safe} to vote based on a locking mechanism (explained later). It sends this vote to the leader. The leader collects $2t+1$ such votes to form an aggregate signature \qctwo. The leader sends \qctwo to all replicas.
    \item[-] \textbf{\commit.} On receiving a \qctwo in view $e$ containing message $v$, a replica locks on $(v, e)$ and sends \commit to the leader. The leader collects $2t+1$ such votes to form an aggregate signature \qcthree. The leader sends \qcthree to all replicas.
    \item[-] \textbf{\reply.} On receiving \qcthree from the leader, replicas output $v$ and send a \reply (along with \qcthree) to the client.
\end{itemize}

Once a replica locks on a value $v$ in view $e$, we call $(v, e)$ is the current lock of this replica. And a higher lock is a lock formed in a higher view. With lock $(v,e)$, the replica only votes for the value $v$ in subsequent views. The only scenario in which it votes for a value $v' \neq v$ is when the status certificate $M$ provides sufficient information stating that $2t+1$ replicas are not locked on $v$. At the end of a view, every replica sends its lock to the leader of the next view. The next view leader collects $2t+1$ such values as a status certificate $M$.

The safety of PBFT follows from two key quorum intersection arguments:

\myparagraph{Uniqueness within a view.} Within a view, safety is ensured by votes in either round. Since a replica only votes once for the first valid value it receives, by a quorum intersection argument, two conflicting values cannot both obtain \qcthree when $f \leq t$. 

\myparagraph{Safety across views.} Safety across views is ensured by the use of locks and the status certificate. First, observe that if a replica $r$ outputs a value $v$ in view $e$, then a quorum of replicas lock on $(v, e)$. When $f \le t$, this quorum includes a set $H$ of at least $t+1$ honest replicas. For any replica in $H$ to update to a higher lock, they need a \qctwo in a higher view $e' > e$, which in turn requires a vote from at least one of these honest replicas in view $e'$. 
However, replicas in $H$ will vote for a conflicting value $v'$ in a higher view only if it is accompanied by a status certificate $M$ containing $2t+1$ locks that are not on value $v$. When $f \le t$, honest replicas in $M$ intersect with honest replicas in $H$  at least one replica -- this honest replica will not vote for a conflicting value $v'$.


\subsection{Forensic Analysis for PBFT-PK}

The agreement property for PBFT holds only when $f \leq t$. When the number of faults are higher, this agreement property (and even termination) can be violated. In this section, we show how to provide forensic support for PBFT when the agreement property is violated. We show that, if two honest replicas output conflicting values $v$ and $v'$ due to the presence of $t<f \le m$ Byzantine replicas, our forensic protocol can detect $t+1$ Byzantine replicas with an irrefutable proof. For each of the possible scenarios in which safety can be violated, the proof shows exactly what property of PBFT was not respected by the Byzantine replicas. The irrefutable proof explicitly uses messages signed by the Byzantine parties, and is thus only applicable to the variant PBFT-PK where all messages are signed. 

\myparagraph{Intuition.} In order to build intuition, let us assume $n = 3t+1$ and $f = t+1$ and start with a simple scenario: two honest replicas output values $v$ and $v'$ in the same view. It must then be the case that a \qcthree is formed for both $v$ and $v'$. Due to a quorum intersection argument, it must be the case that all replicas in the intersection have voted for two conflicting values to break the uniqueness property. Thus, all the replicas in the intersection are culpable. For PBFT-PK, the \qcthree (as well as \qctwo) for the two conflicting values act as the irrefutable proof for detecting $t+1$ Byzantine replicas.

\begin{figure}[htbp]
    \centering
    \includegraphics[width=\columnwidth]{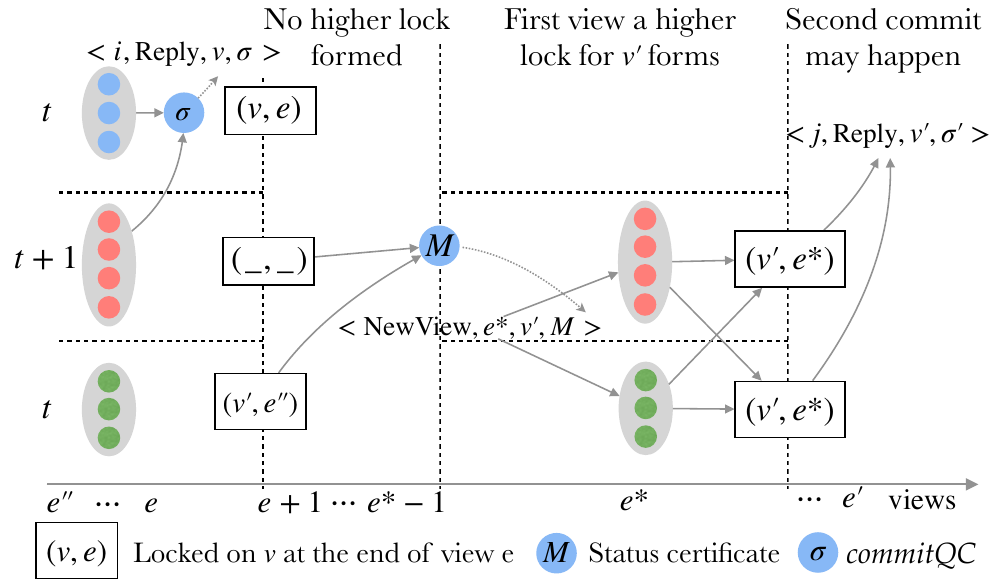}
    \caption{\textbf{An example sequence of events in the PBFT-PK protocol that leads to replicas $i$ and $j$ outputting different values. }}
    \label{fig:pbft-disagree}
\end{figure}

When two honest replicas output conflicting values in different views, there are many different sequences of events that could lead to such a disagreement. One such sequence is described in Figure~\ref{fig:pbft-disagree}. The replicas are split into three sets: the blue set and the green set are honest replicas each of size $t$  and the red replicas are Byzantine replicas of size $t+1$. 

\begin{itemize}[leftmargin=*, topsep=0pt]
    \item In view $e$, replica $i$ outputs $v$ due to $\qcthree$ formed with the $\commit$ from the honest blue set and the Byzantine red set. At the end of the view, replicas in the blue and red set hold locks $(v, e)$ whereas the green set holds a lower lock for a different value. 
    \item In the next few views, no higher locks are formed. Thus, the blue and the red set still hold locks $(v, e)$.
    \item Suppose $e^*$ is the first view where a higher lock is formed. At the start of this view, the leader receives locks from the honest green set who hold lower-ranked locks and the Byzantine red set who maliciously send lower-ranked locks. The set of locks received by the leader is denoted by $M$ and suppose the highest lock was received  for $v'$. The leader proposes $v'$ along with $M$. This can make any honest replica ``unlock'' and vote for $v'$ and form quorum certificates in this view.
    \item In some later view $e'$, replica $j$ outputs $v'$.
\end{itemize}

With this sequence of events, consider the following questions: (1) Is this an admissible sequence of events? (2) How do we find the culpable Byzantine replicas? What does the irrefutable proof consist of? (3) How many replica transcripts do we need to construct the proof?

To answer the first question, the only nontrivial part in our example is the existence of a view $e^*$ where a higher lock is formed. However, such a view $e < e^* \leq e'$ must exist because replica $j$ outputs in view $e'$ and a higher lock must be formed at the latest in view $e'$. 

For the second question, observe that both the red replicas as well as the green replicas sent locks lower than $(v, e)$ to the leader in $e^*$. However, only the red replicas also sent \commit messages for value $v$ in view $e$. Thus, by intersecting the set of \commit messages for value $v$ in view $e$ and the messages forming the status certificate sent to the leader of $e^*$, we can obtain a culpable set of $t+1$ Byzantine replicas. So the proof for PBFT-PK consists of the \qcthree in $e$ and the status certificate in $e^*$, which indicates that the replicas sent a lower lock in view $e^*$ despite having access to a higher lock in a lower view $e$. 

For the third question, the \nv message containing the status certificate $M$ in view $e^*$ can act as the proof, so only one transcript needs to be stored. 

\begin{algorithm}[htbp]
\begin{algorithmic}[1]
\As{a replica running PBFT-PK}
    \State $Q\gets $ all \nv messages in transcript\label{alg:pbft-3}
    \Upon{$\lr{\conflictacross, e, v, e'}$ from a client}
        \For{$m \in Q$}
            \State $(v'',e'')\gets$ the highest lock in $m.M$
            \If{$(m.e \in(e, e'])\land(v''\neq v)\land(e'' \le e)$} 
                \State send $\lr{\nv, m}$ to the client\label{alg:pbft-13}
            \EndIf
        \EndFor
    \EndUpon
\EndAs{}

\As{a client}
    \Upon{two conflicting \reply messages}\label{alg:pbft-18}
        \If{the two messages are from different views}
            \State $\lr{\reply,e,v,\sigma}\gets$ the message from lower view
            \State $e'\gets$ the view number of \reply from higher view
            \State broadcast $\lr{\conflictacross, e, v, e'}$ \label{alg:pbft-21}
            \State wait for: $\lr{\nv, m}$ s.t. $m.e\in(e,e']\land(v''\neq v)\land(e'' \le e)$ where $(v'',e'')$ is the highest lock in $m.M$. \label{alg:pbft-message-com-1}
            \If{in $m.M$ there are two locks $(e'',v_1,\sigma_1)$ and $(e'',v_2,\sigma_2)$ s.t. $v_1\neq v_2$} \label{alg:pbft-20-withinM}
            \State \textbf{output} $\sigma_1\cap \sigma_2$ 
            \Else
            \State \textbf{output} the intersection of senders in $m.M$ and signers of $\sigma$.\label{alg:pbft-23} 
            \EndIf
        \Else\label{alg:pbft-25}
            \ignore{\State $v, v',e\gets$ two output values and the view number
            \State broadcast $\lr{\conflictwithin, e, v, v'}$
            \State wait for: $qc, qc'\gets$ two \qctwo of output values from possibly different replicas in view $e$.
            \State \textbf{output} $qc.\sigma \cap qc'.\sigma$}
            \State $\lr{\reply,e,v,\sigma}\gets$ first \reply message
            \State $\lr{\reply,e,v',\sigma'}\gets$ second \reply message
            \State \textbf{output} $\sigma \cap \sigma'$\label{alg:pbft-28}
        \EndIf
    \EndUpon
\EndAs{}
\end{algorithmic}
\caption{Forensic protocol for PBFT-PK Byzantine agreement}
\label{alg:fa-pbft1}
\end{algorithm}

\myparagraph{Forensic protocol for PBFT-PK.}
Algorithm~\ref{alg:fa-pbft1} describes the entire protocol to obtain forensic support atop PBFT-PK. For completeness, we also provide a complete description of the PBFT-PK protocol in Algorithm~\ref{alg:single-pbft}. Each replica keeps all received messages as transcripts and maintains a set $Q$ containing all received \nv messages (line~\ref{alg:pbft-3}). If a client observes the replies of two conflicting values, it first checks if two values are output in the same view (line~\ref{alg:pbft-18}). If yes, then any two \qcthree for two different output values can provide a culpability proof for at least $t+1$ replicas (lines~\ref{alg:pbft-25}-\ref{alg:pbft-28}). Otherwise, the client  sends a request for a possible proof between two output views $e, e'$ (lines~\ref{alg:pbft-21}). Each replica looks through the set $Q$ for the \nv message in the smallest view $e^* > e$ such that the status certificate $M$ contains the highest lock $(e'', v'')$ where $v''\ne v$ and $e'' \le e$ and sends it to the client (line~\ref{alg:pbft-13}). If inside $M$ there are conflicting locks in the same view, the intersection of them proves at least $t+1$ culprits (line~\ref{alg:pbft-20-withinM}), otherwise the intersection of $M$ and the \qcthree proves at least $t+1$ culprits (line~\ref{alg:pbft-23}).

The following theorem sharply characterizes the forensic support capability of PBFT-PK. As long as $m\leq 2t$, the best possible forensic support is achieved (i.e., $k=1$ and $d=t+1$). Algorithm~\ref{alg:fa-pbft1} can be used to irrefutably detect $t+1$ Byzantine replicas.  Conversely, if $m > 2t$ then no forensic support is possible (i.e., $k=n-f$ (messages from all honest nodes) and $d=0$).    

\begin{theorem}
\label{thm:pbft}
With $n=3t+1$, when $f>t$, if two honest replicas output conflicting values, PBFT-PK provides $(2t,\ 1,\ t+1)$-forensic support. Further $(2t+1, n-f, d)$-forensic support is impossible with $d>0$. 
\end{theorem}
\begin{proof} We prove the forward part of the theorem below. The converse (impossibility) is proved in \S\ref{sec:proofoftheorem4.1}. 
Suppose the values $v$ and $v'$ are output in views $e$ and $e'$ respectively. 

\myparagraph{Case $e=e'$.} 

\noindent\underline{Culpability.} The quorums \qcthree  for $v$ and \qcthree  for $v'$ intersect in $t+1$ replicas. These $t+1$ replicas should be Byzantine since the protocol requires a replica to vote for at most one value in a view.

\noindent\underline{Witnesses.} Client can obtain the culpability proof based on two \qcthree. No additional communication is needed in this case ($k=0$).

\myparagraph{Case $e\neq e'$.}

\noindent\underline{Culpability.}
If $e \neq e'$, then WLOG, suppose $e < e'$. Since $v$ is output in view $e$, it must be the case that $2t+1$ replicas are locked on $(v, e)$ at the end of view $e$ (if they are honest). Now consider the first view $e < e^* \leq e'$ in which a higher lock $(v'', e^*)$ is formed (not necessarily known to any honest party) where $v'' \neq v$ (possibly $v'' = v'$). Such a view must exist since $v'$ is output in view $e' > e$ and a lock will be formed in at least view $e'$. Consider the status certificate $M$ sent by the leader of view $e^*$ in its \nv message. $M$ must contain $2t+1$ locks; each of these locks must be from view $e'' \le e$, and a highest lock among them is $(v'', e'')$.  

We consider two cases based on whether the status certificate contains two different highest locks: (i) there exist two locks $(v'', e'')$ and $(v''', e'')$ s.t. $v'' \ne v'''$ in $M$. (ii) $(v'', e'')$ is the only highest lock in $M$. For the first case, since two locks are formed in the same view, the two quorums forming the two locks in view $e''$ intersect in $t+1$ replicas. These replicas are Byzantine since they voted for more than one value in view $e$.

For the second case, $(v'', e'')$ is the only highest lock in the status certificate $M$. $M$ intersects with the $2t+1$ signers of \qcthree in view $e$ at $t+1$ Byzantine replicas. These replicas are Byzantine because they had a lock on $v \ne v''$ in view $e \ge e''$ but sent a different lock to the leader of view $e^* > e$.


\noindent\underline{Witnesses.} Client can obtain the proof by storing the \nv message containing the status certificate $M$ in $e^{*}$. Only one witness is needed to provide the \nv message ($k=1$). The status certificate $M$ and the first \qcthree act as the irrefutable proof. 
\end{proof}

\myparagraph{Communication complexity.} In the first branch of the forensic protocol, Algorithm~\ref{alg:fa-pbft1}, the client needs to receive one message from $k=1$ replica and the message size is $(2t+1)(|v|+|sig|)$ where $|v|$ and $|sig|$ stand for the size of a value and an aggregate signature (line~\ref{alg:pbft-message-com-1}). In the second branch, the client doesn't need any message (line~\ref{alg:pbft-25}). Therefore the complexity of the client receiving messages is $O(n(|v|+|sig|))$. Notice that we exclude the communication for learning replica outputs (\reply messages) since that procedure happens \textit{before} the forensic protocol.

\subsection{Impossibility for PBFT-MAC}\label{sec:pbft-mac}
We now show an impossibility for a variant of PBFT proposed in \cite[Section 5]{castro1999practical}. The arguments here also apply to the variant in \cite{castro2002practical}. Compared to \S\ref{subsec:pbft-overview}, the only difference in this variant is (i)  \prepare and \commit messages are authenticated using MACs instead of signatures, and (ii) these messages are broadcast instead of routing them through the leader.


\myparagraph{Intuition.} The key intuition behind the impossibility relies on the absence of digital signatures which were used to ``log'' the state of a replica when some replica $i$ outputs a value. In particular, if we consider the example in Figure~\ref{fig:pbft-disagree}, while $i$ receives $2t+1$ \commit messages for value $ v$, these messages are not signed. Thus, if  $t+1$ Byzantine replicas vote for a different value $v'$, $v'$ can be output by a different replica. 
The absence of a verifiable proof stating the set of replicas that sent a \commit to replica $i$ prevents any forensic analysis. We formalize this intuition below.

\ignore{A broadcast message from replica $x$ carries a vector of MACs, which includes the MAC between $x,y$ for any $y\neq x$. When replica $y$ receives the message, it can only check the MAC between $x,y$ and not between $x,z$ for $z\neq y$. As a result, if $x$ is Byzantine, it can include invalid MACs between $x,z$ for all $z\neq y$ and still make $y$ accept such a {\it partially valid} message. And if $y$ is Byzantine, it can also create such a partially valid message. Therefore, given such a partially valid message, a client cannot distinguish whether $x$ or $y$ is Byzantine. We will use this property on \commit messages to construct two indistinguishable worlds. In one world a Byzantine replica $x$ does broadcast a partially valid \commit message, and in another world a Byzantine replica $y$ claims $x$ has broadcast the message therefore frames honest replica $x$.}

\begin{theorem}
With $n=3t+1$, when $f>t$, if two honest replicas output conflicting values, $(t+1, 2t, d)$-forensic support is impossible with $d>0$ for PBFT-MAC.
\end{theorem}

\begin{proof}
Suppose the protocol provides forensic support to detect $d \geq 1$ replicas with irrefutable proof. To prove this result, we construct two worlds where a different set of $t+1$ replicas are Byzantine in each world but a forensic protocol cannot be correct in both worlds. We fix $f = t+1$, although the arguments will apply for any $f > t$.

Let there be four replica partitions $P,Q,R,\{x\}$. $|P|=|Q|=|R|=t$, and $x$ is an individual replica. In both worlds, the conflicting outputs are presented in the same view $e$. Suppose the leader is a replica from set $R$.

\myparagraph{World~1.} Let $P$ and $x$ be Byzantine replicas in this world. The honest leader from set $R$ in view $e$ proposes $v'$. Parties in $R$, $x$ and $Q$ send \prepare and \commit messages (authenticated with MACs) for value $v'$. Due to partial synchrony, none of these messages arrive at $P$. At the end of view $e$, only $R$ and one replica $q$ in $Q$ receive enough \commit messages and send replies to the client. So the client receives the first set of $t+1$ replies for value $v'$, which contain the same quorum $R, x, Q$.

The Byzantine parties in $P$ and $x$ simulate a proposal from the leader for $v$, and the sending of \prepare and \commit messages within $R, P$ and $x$. The simulation is possible due to the absence of a PKI. At the end of view $e$, $P$ and $x$ obtain enough \commit messages and send replies to the client. Thus, the client receives the second set of $t+1$ replies for value $v$, which contain the same quorum $P, R, x$.
 Client starts the forensic protocol.

During the forensic protocol, Byzantine $P$ and $x$ only present the votes for $v$, forged votes from $R$ as their transcripts. Since $t+1$ parties have output each of $v$ and $v'$, there is a safety violation. Since the protocol has forensic support for $d\geq 1$, using these transcripts, the forensic protocol determines some subset of $P$ and $x$ are culpable.

\myparagraph{World~2.} Let $R$ and $q$ (one replica in $Q$) be Byzantine replicas in this world. The Byzantine leader in view $e$ proposes $v$ to $P, R$ and $x$. They send \prepare and \commit messages (authenticated with MACs) for value $v$. These messages do not arrive at $Q$. At the end of view $e$, parties in $P$ and $x$ output $v$. So the client receives the first set of $t+1$ replies for value $v$, which contain the same quorum $P, R, x$.

Similarly, the leader sends $v'$ to $Q, R$ and $x$. The proposal does not arrive at $x$. Only $Q$ and $R$ send \prepare and \commit messages (authenticated with MACs) for $v'$, these messages do not arrive at $P$. However, $R$ and $q$ forge \prepare and \commit messages from $x$. At the end of view $e$, $R$ and $q$ output $v'$. So the client receives the second set of $t+1$ replies for value $v'$, which contain the same quorum $R, x, Q$. Client starts the forensic protocol.
 
During the forensic protocol, Byzantine $R$ and $q$ sends the same transcripts as in World~1 by dropping votes for $v$ and forging votes from $x$. Again, since $t+1$ parties have output each of $v$ and $v'$, there is a safety violation. However, observe that the transcript presented to the forensic protocol is identical to that in World~1. Thus, the forensic protocol outputs some subset of $P$ and $x$ as culpable replicas. In World~2, this is incorrect since replicas in $P$ and $x$ are honest. This completes the proof.

\end{proof}

\section{Forensic Support for HotStuff}
\label{sec:basic-hotstuff}
HotStuff~\cite{yin2019hotstuff} is a partially synchronous consensus protocol that provides an optimal resiliency of $n = 3t+1$.  The HotStuff protocol is similar to PBFT but there are subtle differences which allow it to obtain a linear communication complexity for both its steady state and view change protocols (assuming the presence of threshold signatures). Looking ahead, these differences significantly change the way forensics is conducted if a safety violation happens.

\subsection{Overview}  We start with an overview of the protocol. For simplicity, we discuss a single-shot version of HotStuff. The protocol proceeds in a sequence of consecutive views where each view has a unique leader. Each view of HotStuff progresses as follows:\footnote{The description of HotStuff protocol is slightly different from the basic algorithm described in~\cite[Algorithm~2]{yin2019hotstuff} to be consistent with the description of PBFT in \S\ref{subsec:pbft-overview}.}
\begin{itemize}[leftmargin=*, topsep=0pt]
    \item[-] \textbf{\ppr.} The leader proposes a \nv message containing a proposal $v$ along with the \hqc (the highest \qctwo known to it) and sends it to all replicas.
    \item[-] \textbf{\prepare.} On receiving a \nv message containing a valid value $v$ in a view $e$ and a \hqc, a replica sends \prepare for $v$ if it is \emph{safe} to vote based on a locking mechanism (explained later). It sends this vote to the leader. The leader collects $2t+1$ votes to form an aggregate signature \qctwo in view $e$. The leader sends the view $e$ \qctwo to all replicas.
    \item[-] \textbf{\precommit.} On receiving a \qctwo in view $e$ containing message $v$, a replica updates its highest \qctwo to $(v, e)$ and sends \precommit to the leader. The leader collects $2t+1$ such votes to form an aggregate signature \qcprecom.
    \item[-] \textbf{\commit.} On receiving \qcprecom in view $e$ containing message $v$ from the leader, a replica locks on $(v, e)$ and sends \commit to the leader. The leader collects $2t+1$ such votes to form an aggregate signature \qcthree.
    \item[-] \textbf{\reply.} On receiving \qcthree from the leader, replicas output the value $v$ and send a \reply (along with \qcthree) to the client.
\end{itemize}

Once a replica locks on a given value $v$, it only votes for the value $v$ in subsequent views. The only scenario in which it votes for a value $v' \neq v$ is when it observes a \hqc from a higher view in a \nv message. At the end of a view, every replica sends its highest \qctwo to the leader of the next view. The next view leader collects $2t+1$ such values and picks the highest \qctwo as \hqc. 
The safety and liveness of HotStuff when $f \leq t$ follows from the following:

\myparagraph{Uniqueness within a view.} Since replicas only vote once in each round, a \qcthree can be formed for only one value when $f \leq t$.

\myparagraph{Safety and liveness across views.} Safety across views is ensured using locks and the voting rule for a \nv message. Whenever a replica outputs a value, at least $2t+1$ other replicas are locked on the value in the view. Observe that compared to PBFT, there is no status certificate $M$ in the \nv message to ``unlock'' a replica. Thus, a replica only votes for the value it is locked on. The only scenario in which it votes for a conflicting value $v'$ is if the leader includes a \qctwo for $v'$ from a higher view in \nv message. This indicates that at least $2t+1$ replicas are not locked on $v$ in a higher view, and hence it should be safe to vote for it. The latter constraint of voting for $v'$ is not necessary for safety, but only for liveness of the protocol.

\myparagraph{Variants of HotStuff.}
In this paper, we study three variants of HotStuff, identical for the purposes of consensus but provide varied forensic support. The distinction among them is only in the information carried in \vone message. For all three versions, the message contains the message type \vone, the current view number $e$ and the proposed value $v$. In addition, \vone in \hsa contains $\eqc$, the view number of the \hqc in the \nv message. \hsb contains the hash of \hqc (cf. Table~\ref{tab:hotstuff_variants}). \hsb is equivalent to the basic algorithm described in~\cite[Algorithm 2]{yin2019hotstuff}. \hsc does not add additional information.

\begin{table}[]
    \centering
    \begin{tabular}{c|c|c|c}
    \toprule
         & \hsa &\hsb&\hsc  \\
        \toprule
        \info & $\eqc$&$ \text{Hash}(\hqc)$&$\emptyset$\\
        $m$ & $2t$ & $2t$ &  $t+1$\\

        $k$ & $1$ & $t+1$ & $2t$\\
        $d$ & $t+1$ & $t+1$ & $1$\\
        \bottomrule
    \end{tabular}
    \caption{Comparison of different variants of HotStuff, the \prepare message is $\signed{\vone,e,v, \info}$}
    \label{tab:hotstuff_variants}
\end{table}

\subsection{Forensic Analysis for HotStuff}
If two conflicting values are output in the same view, Byzantine replicas can be detected using \qcthree  and using ideas similar to that in PBFT. However, when the conflicting outputs of replicas $i$ and $j$ are across views $e$ and $e'$ for $e < e'$, the same ideas do not hold anymore. To understand this, observe that the two key ingredients for proving the culpability of Byzantine replicas in PBFT were (i) a \qcthree  for the value output in a lower view (denoted by $\sigma$ for replica $i$'s reply in Figure~\ref{fig:pbft-disagree}) and (ii) a status certificate from the first view higher than $e$ containing the locks from $2t+1$ replicas (denoted by $M$ for view $e^* > e$ in Figure~\ref{fig:pbft-disagree}). In HotStuff, a \qcthree still exists. However, for communication efficiency reasons, HotStuff does not include a status certificate $M$ in its \nv message. The status certificate in PBFT provides us with the following:

\begin{itemize}[leftmargin=*, topsep=0pt]
    \item \textbf{Identifying a potential set of culpable replicas.} Depending on the contents of $M$ and knowing $\sigma$, we could identify a set of Byzantine replicas.
    \item \textbf{Determining whether the view is the first view where a higher lock for a conflicting value is formed.} By inspecting all locks in $M$, we can easily determine this. Ensuring first view with a higher lock is important; once a higher lock is formed, even honest replicas may update their locks and the proof of Byzantine behavior may not exist in the messages in subsequent views.
\end{itemize}

\begin{figure}[tbp]
    \centering
    \includegraphics[width=\columnwidth]{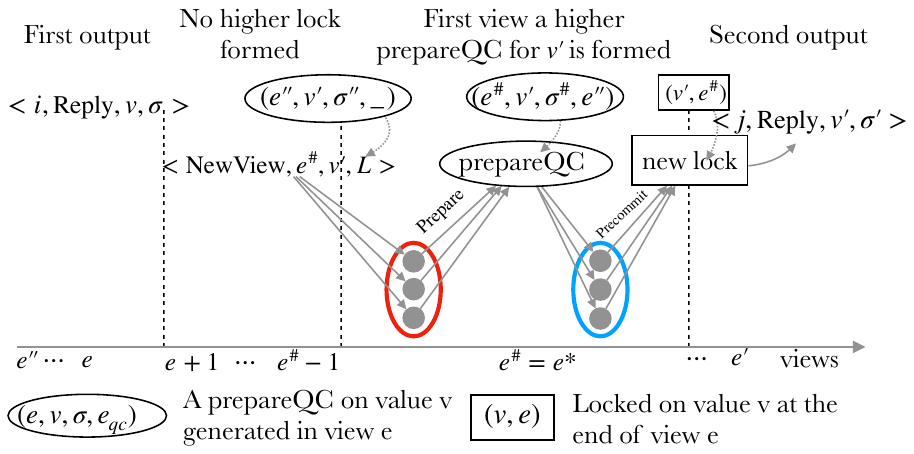}
    \caption{\textbf{Depiction of events in the \hsa protocol for the first view where a higher \qctwo for $v'$ is formed.} } 
    \label{fig:hotstuff-same}
\end{figure}

Let us try to understand this based on the first view $e^\#$ where a higher \qctwo is formed for $v'\ne v$ (see Figure~\ref{fig:hotstuff-same}).  The set of replicas who sent \vone (the red ellipse) in $e^\#$ and formed a \qctwo are our potential set of Byzantine replicas. Why? If $e^\#$ is indeed the first view in which a higher \qctwo is formed, then all of these replicas voted for a \nv message containing a \hqc from a lower or equal view $e''$ on a different value. If any of these replicas also held a lock $(v, e)$ (by voting for replica $i$'s output) then these replicas have output the culpable act of not respecting the voting rule. 

The only remaining part is to ensure that this is indeed the first view where a higher conflicting \qctwo is formed. The way to prove this is also the key difference among three variants of HotStuff. For \hsa, \qctwo contains $\eqc$, which explicitly states the view number of \hqc in the \nv message they vote for. If $\eqc < e$, \qctwo provides an irrefutable proof for culpable behavior. For \hsb, the hash information contained in \vone provides the necessary link to the \nv message they vote, so once the linked \nv message is accessible, the \qctwo and \nv together serve as the proof for culpable behavior. However for \hsc, even if we receive both \qctwo and \nv messages that are formed in the same view, no proof can be provided to show a connection between them. A Byzantine node vote for the first higher \qctwo can always refuse to provide the \nv message they receive.

Thus, to summarize, the red set of replicas in view $e^\#$ are a potential set of culpable nodes of size $t+1$. The irrefutable proof to hold them culpable constitutes two parts, (1) the first \qctwo containing their signed \vone messages, and (2) a proof to show this is indeed the first view. In the next two subsections we will introduce the forensic protocols for \hsa and \hsb to provide different forensic supports, and how it is impossible for \hsc to provide forensic support.

\subsection{Forensic Protocols for \hsa and \hsb}

\begin{algorithm}[htbp]
\begin{algorithmic}[1]
\As{a replica running \hsa}
    \State $P\gets $ all \qctwo in transcript \Comment{including \qctwo in \precommit message and \hqc in \nv message}\label{alg:hotstuffa-2}
    \Upon{$\lr{\conflictacross, e, v, e'}$ from a client}
        \For{$qc\in P$}
            \If{$(qc.v\neq v)\land (qc.e\in (e,e'])\land (qc.\eqc\le e)$} 
                \State send $\lr{\proofacross, qc}$ to the client\label{alg:hotstuffa-12}
            \EndIf
        \EndFor
    \EndUpon
\EndAs{}

\As{a client}
    \Upon{two conflicting \reply messages}
        \If{ two messages are from different views}\label{hotstuffa:28}
            \State $\lr{\reply,e,v,\sigma}\gets$ the message from lower view
            \State  $e'\gets$ the view number of message from higher view
            \State broadcast $\lr{\conflictacross, e, v, e'}$ \label{alg:hotstuffa-23}
            
            \State wait for  $\lr{\proofacross, qc}$ s.t.
            \begin{enumerate}[leftmargin=1in]
                \item[(1)] $e<qc.e\le e'$, and\label{alg:hotstuffa-24}
                \item[(2)] $(qc.v\neq v)\land (qc.\eqc\le e)$
            \end{enumerate}
            \State \textbf{output} $qc.\sigma \cap \sigma$\label{alg:hotstuffa-26}
        \Else \label{alg:hotstuffa-28}
            \State $\lr{\reply,e,v,\sigma}\gets$ first \reply message
            \State $\lr{\reply,e,v',\sigma'}\gets$ second \reply message
            \State \textbf{output} $\sigma \cap \sigma'$
        \EndIf
    \EndUpon
\EndAs{}
\end{algorithmic}
\caption{Forensic protocol for \hsa}
\label{alg:fa-hotstuff-a}
\end{algorithm}

\begin{algorithm}[htbp]
\begin{algorithmic}[1]
\As{a replica running \hsb}
    \State $P\gets $ all \qctwo in transcript\label{alg:hotstuffb-2}
    \State $Q\gets $ all \nv messages in transcript\label{alg:hotstuffb-3}
    \Upon{$\lr{\conflictacross, e, v, e'}$ from a client}
        \For{$qc\in P$}
            \If{$(qc.v\neq v)\land (qc.e\in (e,e'])$} 
                \State send $\lr{\proofacross, qc}$ to the client\label{alg:hotstuffb-12}
            \EndIf
        \EndFor
        \For{$m \in Q$}
            \If{$(m.v\neq v)\land(m.e\in (e,e'])\land(m.\hqc.e \le e)$} 
                \State send $\lr{\nv, m}$ to the client\label{alg:hotstuffb-13}
            \EndIf
        \EndFor
    \EndUpon
\EndAs{}

\As{a client}
    \State $NV\gets \{\}$
    \Upon{two conflicting \reply messages}
        \If{ two messages are from different views}\label{hotstuffb:28}
            \State $\lr{\reply,e,v,\sigma}\gets$ the message from lower view
            \State $e'\gets$ the view number of message from higher view
            \State broadcast $\lr{\conflictacross, e, v, e'}$ \label{alg:hotstuffb-23}
            \Upon{$\lr{\nv, m}$}
                \If{$(m.v\neq v)\land(m.e\in (e,e'])\land(m.\hqc.e \le e)$} 
                    \State $NV\gets NV\cup \{m\}$
                \EndIf
            \EndUpon
            \State wait for  $\lr{\proofacross, qc}$ s.t.
            \begin{enumerate}[leftmargin=1in]
                \item[(1)] $e<qc.e\le e'$, and\label{alg:hotstuffb-24}
                \item[(2)] $(qc.v\neq v)\land  (\exists m\in NV, \text{Hash}(m) = qc.hash)$
            \end{enumerate}
            \State \textbf{output} $qc.\sigma \cap \sigma$\label{alg:hotstuffb-26}
        \Else \label{alg:hotstuffb-28}
            \State $\lr{\reply,e,v,\sigma}\gets$ first \reply message
            \State $\lr{\reply,e,v',\sigma'}\gets$ second \reply message
            \State \textbf{output} $\sigma \cap \sigma'$
        \EndIf
    \EndUpon
\EndAs{}
\end{algorithmic}
\caption{Forensic protocol for \hsb}
\label{alg:fa-hotstuff-b}
\end{algorithm}

\myparagraph{Forensic protocol for \hsa.}
Algorithm~\ref{alg:fa-hotstuff-a} describes the protocol to obtain forensic support atop \hsa. A complete description of the general HotStuff protocol is also provided in Algorithm~\ref{alg:single-hotstuff-modified}. Each replica keeps all received messages as transcript and maintains a set $P$ containing all received \qctwo from \precommit messages and \hqc from \nv messages (line~\ref{alg:hotstuffa-2}). If a client observes outputs of two conflicting values in the same view, it can determine the culprits using the two \reply messages (line~\ref{alg:hotstuffa-28}).  Otherwise, the client sends a request to all replicas for a possible proof between two output views $e, e'$ for $e<e'$ (line~\ref{alg:hotstuffa-23}). Each replica looks through the set $P$  for \qctwo formed in views $e<e^\#\le e'$. If there exists a \qctwo whose value is different from the value $v$ output in $e$ and whose $\eqc$ is less than or equal to $e$, it sends a reply with this \qctwo to the client (line~\ref{alg:hotstuffa-12}). The client waits for a \qctwo (line~\ref{alg:hotstuffa-24}) formed between two output views. For \hsa, if it contains a different value from the first output value and an older view number $\eqc < e$, the intersection of this \qctwo and the \qcthree from the \reply message in the lower view proves at least $t+1$ culprits (line~\ref{alg:hotstuffa-26}).

\myparagraph{Forensic protocol for \hsb.} 
Algorithm~\ref{alg:fa-hotstuff-b} describes the protocol to obtain forensic support atop \hsb, which is similar to the protocol for \hsa. For replicas running \hsb, besides $P$, they also maintains the set $Q$ for received \nv messages (line~\ref{alg:hotstuffb-3}). When receiving a forensic request from clients, replicas look through $P$ for \qctwo formed in views $e<e^\#\le e'$ and send all \qctwo whose values are different from the value $v$ to the client (line~\ref{alg:hotstuffb-12}). Besides, they also look through $Q$ for a \nv message formed in views $e<e^\#\le e'$ and send all \nv proposing a value different from $v$ and containing a \hqc with view number $\le e$ (line~\ref{alg:hotstuffb-13}). For \hsb, when receiving such a \nv for different values, the message will be stored temporarily by the client until a \qctwo for the \nv message with a matching hash is received. The \nv and the \qctwo together form the desired proof; the intersection of the \qctwo and the \qcthree provides at least $t+1$ culprits.

\myparagraph{Forensic proofs.} The following theorems characterize the forensic support capability of \hsa and \hsb. As long as $m\leq 2t$, \hsa can achieve the best possible forensic support (i.e., $k=1$ and $d=t+1$). \hsb can achieve a medium forensic support (i.e., $k=t+1$ and $d=t+1$). Conversely, if $m > 2t$ then no forensic support is possible for both protocols (i.e., $k=n-f$ and $d=0$).

\begin{theorem}
\label{thm:hotstuff_a}
With $n=3t+1$, when $f>t$, if two honest replicas output conflicting values, \hsa provides $(2t,\ 1,\ t+1)$-forensic support. Further $(2t+1, n-f, d)$-forensic support is impossible with $d>0$.
\end{theorem}
\begin{proof}We prove the forward part of the theorem below. The proof of converse (impossibility) is the same as \S\ref{sec:proofoftheorem4.1}.
Suppose two conflicting values $v,\ v'$ are output in views $e$, $e'$ respectively. 

\myparagraph{Case $e=e'$.}

\noindent\underline{Culpability.} The \qcthree of $v$ and \qcthree of $v'$ intersect in $t+1$ replicas. These $t+1$ replicas should be Byzantine since the protocol requires a replica to vote for at most one value in a view. 

\noindent\underline{Witnesses.} Client can obtain a proof based on the two \reply messages, so additional witnesses are not necessary in this case.

\myparagraph{Case $e\neq e'$.}

\noindent\underline{Culpability.}
If $e \neq e'$, then WLOG, suppose $e < e'$. Since $v$ is output in view $e$, it must be the case that $2t+1$ replicas are locked on $(v, e)$ at the end of view $e$. Now consider the first view $e < e^* \leq e'$ in which a higher lock $(v'', e^*)$ is formed where $v''\ne v$ (possibly $v''=v'$). Such a view must exist since $v'$ is output in view $e' > e$ and a lock will be formed in at least view $e'$. For a lock to be formed, a higher \qctwo must be formed too. Consider the first view $e<e^\# \leq e'$ in which the corresponding \qctwo for $v''$ is formed. The leader in $e^\#$ broadcasts the \nv message containing a \hqc on $(v'', e'')$. Since this is the first time a higher \qctwo is formed and there is no higher \qctwo for $v''$ formed between view $e$ and $e^\#$, we have $e''\le e$. The formation of the higher \qctwo indicates that $2t+1$ replicas received the \nv message proposing $v''$ with \hqc on $(v'', e'')$ and consider it a valid proposal, i.e., the view number $e''$ is larger than their locks because the value is different. 

 Recall that the output value $v$ indicates $2t+1$ replicas are locked on $(v, e)$ at the end of view $e$. In this case, the $2t+1$ votes in \qctwo in view $e^\#$ intersect with the $2t+1$ votes in \qcthree in view $e$ at $t+1$ Byzantine replicas. These replicas should be Byzantine because they were locked on the value $v$ in view $e$ and vote for a value $v'' \neq v$ in a higher view $e^\#$ when the \nv message contained a \hqc from a  view $e'' \leq e$. Thus, they have violated the voting rule.

\noindent\underline{Witnesses.} Client can obtain a proof by storing a \qctwo formed between $e$ and $e'$, whose value is different from $v$ and whose $\eqc \le e$. So only one witness is needed ($k=1)$, the \qctwo and the first \qcthree act as the irrefutable proof.
\end{proof}

\begin{theorem}
\label{thm:hotstuff_b}
With $n=3t+1$, when $f>t$, if two honest replicas output conflicting values, \hsb provides $(2t,\ t+1,\ t+1)$-forensic support. Further $(2t+1, n-f, d)$-forensic support is impossible with $d>0$. 
\end{theorem}

\begin{proof}We prove the forward part of the theorem below. The proof of converse (impossibility) is the same as \S\ref{sec:proofoftheorem4.1}. Suppose two conflicting values $v,\ v'$ are output in views $e$, $e'$ respectively.

\myparagraph{Case $e=e'$.}
Same as Theorem~\ref{thm:hotstuff_a}.

\myparagraph{Case $e\neq e'$.}

\noindent\underline{Culpability.} 
Same as Theorem~\ref{thm:hotstuff_a}.

\noindent\underline{Witnesses.}
Since \qctwo of \hsb only has the hash of \hqc, the irrefutable proof contains the \nv message that includes the \hqc and the corresponding \qctwo with the matching hash $\text{Hash}(\hqc)$. The client may need to store all \nv messages between $e$ and $e'$ whose value is different from $v$ and the whose $\hqc$ is formed in $\eqc \le e$, until receiving a \qctwo for some \nv message with a matching hash. In the best case, some replica sends both the \nv message and the corresponding \qctwo so the client only needs to store $k=1$ replica's transcript. In the worst case, we can prove that any $t+1$ messages of transcript are enough to get the proof. Consider the honest replicas who receive the first \qctwo and the \nv message. $2t+1$ replicas have access to the \qctwo and $2t+1$ replicas have access to the \nv message. Among them at least $t+1$ replicas have access to both messages, and we assume they are all Byzantine. Then at least $t$ honest replicas have the \qctwo and at least $t$ honest replicas have the \nv message. The total number of honest replicas $n-f \le 2t$. Thus among any $t+1$ honest replicas, at least one have \nv message and at least one have \qctwo. Therefore, $t+1$ transcripts from honest replicas ensure the access of both \nv message and \qctwo and thus guarantee the irrefutable proof. 

\end{proof}

\myparagraph{Communication complexity.} In line~\ref{alg:hotstuffa-24} of Algorithm~\ref{alg:fa-hotstuff-a}, the client needs to receive one message from $k=1$ replica and the message size is $(|v|+|sig|)$ where $|v|$ and $|sig|$ stand for the size of a value and an aggregate signature. Therefore the complexity of the client receiving messages is $O(|v|+|sig|)$ for \hsa. As for \hsb, theorem~\ref{thm:hotstuff_b} shows that in the worst case, the client needs to receive messages from $k=t+1$ replicas. Each of those replicas sends one message of size $O(|v|+|sig|+|hash|)$ where $|hash|$ stands for the size of a hash value. Therefore the complexity of the client receiving messages is $O(n(|v|+|sig|+|hash|))$ for \hsb.

\subsection{Impossibility for \hsc}
Compared to the other two variants, in \hsc, \vone message and \qctwo are not linked to the \nv message. We show that this lack of information is sufficient to guarantee impossibility of forensics.

\myparagraph{Intuition.}
When $f=t+1$, from the forensic protocols of \hsa and \hsb, we know that given across-view $\qcthree_1$ and $\qcthree_2$ (ordered by view) and the first \qctwo higher than $\qcthree_1$, the intersection of \qctwo and $\qcthree_1$ contains at least $d=t+1$ Byzantine replicas. The intersection argument remains true for \hsc, however, it is impossible for a client to decide whether \qctwo is the first one only with the transcripts sent by $2t$ honest replicas (when $f=t+1$). In an execution where there are two \qctwo in view $e^*$ and $e'$ respectively ($e^*<e'$), 
the Byzantine replicas (say, set $P$) may not respond with the \qctwo in $e^*$. The lack of information disallows a client from separating this world from another world $P$ is indeed honest and sharing all the information available to them.  We formalize this intuition in the theorem below. 

\begin{theorem}
With $n=3t+1$, when $f>t$, if two honest replicas output conflicting values, $(t+1,\ 2t,\ d)$-forensic support is impossible with $d>1$ for \hsc.
 Further, $(t+2,\ n-f,\ d)$-forensic support is impossible with $d>0$.
\label{thm:hotstuff-null}
\end{theorem}
The theorem is proved in \S\ref{proof:hotstuff-null}.

\section{Forensic Support for Algorand}
\label{sec:algorand}
Algorand~\cite{chen2019algorand} is a synchronous consensus protocol which tolerates up to one-third fraction of Byzantine users. At its core, it uses a BFT protocol from \cite{micali2018byzantine,chen2018algorand}. However, Algorand runs the protocol by selecting a small set of replicas, referred to as the committee, thereby achieving consensus with sub-quadratic communication. 
The protocol is also player replaceable, i.e., different steps of the protocol have different committees, thus tolerating an adaptive adversary. Each replica uses {\it cryptographic self-selection} using a verifiable random function (VRF)~\cite{micali1999verifiable} to privately decide whether it is selected in a committee. The VRF is known only to the replica until it has sent a message to other parties thus revealing and proving its selection in the committee. In this section, we present an overview of the BFT protocol and then show why it is impossible to achieve forensic support for this protocol. 

\subsection{Overview}
We start with an overview of the single-shot version of Algorand~\cite{chen2019algorand}. The protocol assumes synchronous communication, where messages are delivered within a known bounded time. The protocol proceeds in consecutive steps, each of which lasts for a fixed amount of time that guarantees message delivery. Each step has a self-selected committee, and a replica can compute its VRF, a value that decides whether it is selected in the committee, and it is known only by the replica itself until it has sent the value. All messages sent by a committee member is accompanied by a VRF thus allowing other replicas to verify its inclusion in the committee. Parameters such as committee size $\kappa$ are chosen such that the number of honest parties in the committee is greater than a threshold $t_H \geq 2\kappa/3$ with overwhelming probability. 

The BFT protocol is divided into two sub-protocols: \gradedconsensus and \bba. In \gradedconsensus, which forms the first three steps of the protocol, each replica $r$ inputs its value $v_r$. Each replica $r$ outputs a tuple containing a value $v^{\text{out}}_r$ and a grade $g_r$. In an execution where all replicas start with the same input $v$, $v^{\text{out}}_r = v$ and $g_r = 2$ for all replicas $r$. 
The replicas then enter the next sub-protocol, denoted \bba. If $g_r = 2$, replica $r$ inputs value $b_r = 0$, otherwise it inputs $b_r = 1$. At the end of \bba, the replicas agree on the tuple $(0, v)$ or $(1, v_\perp)$.\footnote{$v_\perp$ is considered external valid in Algorand.} The \bba sub-protocol also uses a random coin; for simplicity, we assume access to an ideal global random coin. Our forensic analysis in the next section only depends on \bba and thus, we only provide an overview for \bba here. The protocol proceeds in the following steps,

\begin{itemize}
    \item Steps 1-3 are \gradedconsensus. At the end of \gradedconsensus, each replica inputs a value $v_r$ and a binary value $b_r$ to \bba. 
    \item Step 4. Each replica in the committee broadcasts its vote for $(b_r, v_r)$ along with its VRF.
    \item Step $s$ ($s \geq 5, s \equiv 2 \mod 3$) is the \coinfixtozero step of \bba. In this step, a replica checks {\it Ending Condition 0}: if it has received $\ge t_H$ valid votes on $(b,v)$ from the previous step, where $b=0$, it outputs $(b,v)$ and ends its execution. Otherwise, it updates $b_r$ as follows: 
    \[b_r\gets \left\{\begin{array}{lc}
         1,& \text{if}\ge t_H\text{ votes on $b=1$} \\
         0,& \text{otherwise}
    \end{array}\right.\]
    If the replica is in the committee based on its VRF, it broadcasts its vote for $(b_r,v_r)$ along with the VRF.
    \item Step $s$ ($s \geq 6, s \equiv 0 \mod 3$). Symmetric to the previous step but for bit $1$ instead of $0$. Also, the votes need not be for the same $v$ in the ending condition. 
    \item Steps $s$ ($s\ge7, s\equiv 1\mod 3$) is the \coingenuinelyflipped step of \bba. In this step, it updates its variable $b_r$ as follows:
    \[b_r\gets \left\{\begin{array}{l}
         0, \text{if}\ge t_H\text{ votes on $b=0$} \\
         1, \text{if}\ge t_H\text{ votes on $b=1$} \\
         \text{random coin of step }s,\quad \text{otherwise}
    \end{array}\right.\]
    If the replica is in the committee based on its VRF, it broadcasts its vote for $(b_r,v_r)$ along with the VRF.
\end{itemize}

\myparagraph{Safety of \bba within a step.} If all honest replicas reach an agreement before any step, the agreement will hold after the step. If the agreement is on binary value 0 (1 resp.) then the opposite Ending Condition 1 (0 resp.) will not be satisfied during the step. This is because synchronous communication ensures the delivery of at least $t_H$ votes on the agreed value and there are not enough malicious votes on the other value. 

\myparagraph{Safety of \bba across steps.} For the step \coinfixtozero (1 resp.), if any honest replica ends due to Ending Condition 0 (1 resp.), all honest replicas will agree on binary value 0 and value $v$ (1 and $v_\perp$ resp.) at the end of the step, because there could only be less than $t_H$ votes on binary value 1 (0 resp.). Hence, together with safety within a step, binary value 1 and value $v_\perp$ (0 and $v$ resp.) will never be output.

\subsection{Impossibility of Forensics}
When the Byzantine fraction in the system is greater than one-third, with constant probability, a randomly chosen committee of size $\kappa < n$ will have $t_H < 2\kappa/3$. In such a situation, we can have a safety violation. Observe that since only $\kappa < n$ committee members send messages in a round, the number of culpable replicas may be bounded by $O(\kappa)$. However, we will show an execution where no Byzantine replica can be held culpable. 

\myparagraph{Intuition.} The safety condition for \bba relies on the following: if some honest replica commits to a value $b$, say $b = 0$, in a step and terminates, then all honest replicas will set $b=0$ as their local value. In all subsequent steps, there will be sufficient ($> 2/3$ fraction) votes for $b=0$ due to which replicas will never set their local value $b=1$. Thus, independently of what Byzantine replicas send during the protocol execution, honest replicas will only commit on $b = 0$. On the other hand, if replicas do not receive $> 2/3$ fraction of votes for $b=0$, they may switch their local value to $b=1$ in the \coinfixtoone or \coingenuinelyflipped step. This can result in a safety violation. When the Byzantine fraction is greater than one-third, after some replicas have committed 0, the Byzantine replicas can achieve the above condition by selectively not sending votes to other replicas (say set $Q$), thereby making them switch their local value to $b=1$. There is no way for an external client to distinguish this world from another world where the set $Q$ is Byzantine and states that it did not receive these votes. We formalize this intuition in the theorem below. Observe that our arguments work for the \bba protocol with or without player-replaceability. 


\begin{theorem}
When the Byzantine fraction exceeds 1/3, if two honest replicas output conflicting values, $(t+1,\ 2t,\ d)$-forensic support is impossible with $d>0$ for Algorand.
\label{thm:algorand}
\end{theorem}
The theorem is proved in \S\ref{sec:algorand-proof}.
\section{LibraBFT and {\sf Diem}}
\label{sec:impl}
In this paper, we have focused on forensics for single-shot consensus. Chained BFT protocols are natural candidates for consensus on a sequence with applications to  blockchains. LibraBFT is a chained version of Hotstuff and is the core consensus protocol in {\sf Diem}, a new cryptocurrency supported by Facebook~\cite{diem}. In this section, we show that LibraBFT has the strongest forensic support possible (as in Hotstuff-view). Further, we implement the corresponding forensic analysis protocol as a module on top of an open source {\sf Diem} client. We highlight the system innovations of our implementation and associated forensic dashboard; this has served as a reference implementation  presently under active consideration for deployment (due to anonymity imperativs we are unable to document this more concretely). 

\myparagraph{{\sf Diem} blockchain.}
{\sf Diem} Blockchain uses LibraBFT~\cite{librabft}, 
 a variant of HotStuff protocol for consensus. In LibraBFT, the replicas are called validators, who receive transactions from clients and propose blocks of transactions in a sequence of rounds. 

\myparagraph{LibraBFT forensics.} The culpability analysis for LibraBFT is similar to Theorem~\ref{thm:hotstuff_b}. However, for the witnesses, the blockchain property of LibraBFT makes sure that any replica (validator) has access to the full blockchain and thus provides $(n-2, 1, t+1)$-forensic support. The formal result is below (proof in  \S\ref{sec:proof-diem}).
\begin{theorem}
\label{thm:diem}
For $n=3t+1$, when $f>t$, if two honest replicas output conflicting blocks, LibraBFT provides $(n-2,\ 1,\ t+1)$-forensic support.
\end{theorem}

The aforementioned three variants of HotStuff in \S\ref{sec:basic-hotstuff} are described under the VBA setting to reach consensus on a single value, and the value can be 
 directly included in the vote message (cf.\  Table~\ref{tab:ba_chain}, $v$ is contained in the \prepare message). In this setting, once a replica receives the \qcthree for the value, it will output the value and send a reply to the client, even if the \qcthree is the only message it receives in the current view so far. So when two honest replicas output conflicting values, it is possible that the client receives only the commit messages and extra communication is needed. And when $m > 2t$, Byzantine replicas are able to form QCs by themselves so that no other honest replicas can get access to the first \qctwo. Thus the bound on $m$ for \hsa and \hsb is $2t$.

However, the setting is slightly different in practice, when the value $v$ is no longer a single value, but actually a block with more fields and a list of transactions/commands, as in  LibraBFT. In single-shot consensus,   a block includes the transactions (value) and the \hqc. In this case, the block is too heavy to include directly in a vote message, 
so the replicas add the hash of the block  to the vote message (see  Table~\ref{tab:ba_chain}, Hash$(b=(v, \hqc))$ is contained in the \prepare message). And since only the \nv message has the block's preimage, replicas cannot vote/output until receiving the original blocks. Thus when two honest replicas output conflicting values, the client can obtain the full blockchain from one of them ($k=1$) and all \qctwo are part of the blocks. In this case, even if $m > 2t$ the client can still enjoy non-trivial forensic support. 

\begin{table}[]
\begin{tabular}{c|c|c}
\toprule
     & \hsb  & LibraBFT  \\ \toprule
\vone                                                       & \begin{tabular}[c]{@{}c@{}}$\lr{\vone, e,v,$ \\ $ \text{Hash}(\hqc)}$\end{tabular} & \begin{tabular}[c]{@{}c@{}}$\lr{\vone, e, $\\ $\text{Hash}(b=(v, \hqc))}$\end{tabular} \\ \midrule

$k$  & $t+1$    & $1$   \\ \midrule
$m$   & $2t$    & $n-2$     \\ \midrule
\begin{tabular}[c]{@{}c@{}}extra\\ condition\end{tabular} & - & \begin{tabular}[c]{@{}c@{}}must receive the\\ preimage of hash\end{tabular} \\ \bottomrule
\end{tabular}
    \caption{Comparison of \hsb and LibraBFT 
    }
    \label{tab:ba_chain}
\end{table}

\myparagraph{Forensic module.} Our  prototype consists of two components, a database \fs used to store quorum certificates received by validators, which can be accessed by clients through JSON-RPC requests or consensus API; and an independent \detector run by clients to analyze the forensic information. 
\begin{itemize}
    \item \textbf{\fs} maintains a map from the view number to quorum certificates and its persistent storage. 
    It is responsible for storing forensic information and allows access by other components,  
    including clients (via JSON-RPC requests or consensus API).  
    \item \textbf{\detector } is run by clients manually to send requests periodically to connected validators. It collates information received from validators, using it as the  input to the forensic analysis protocol.
\end{itemize}

\begin{figure}
    \centering
    \includegraphics[width=0.9\columnwidth]{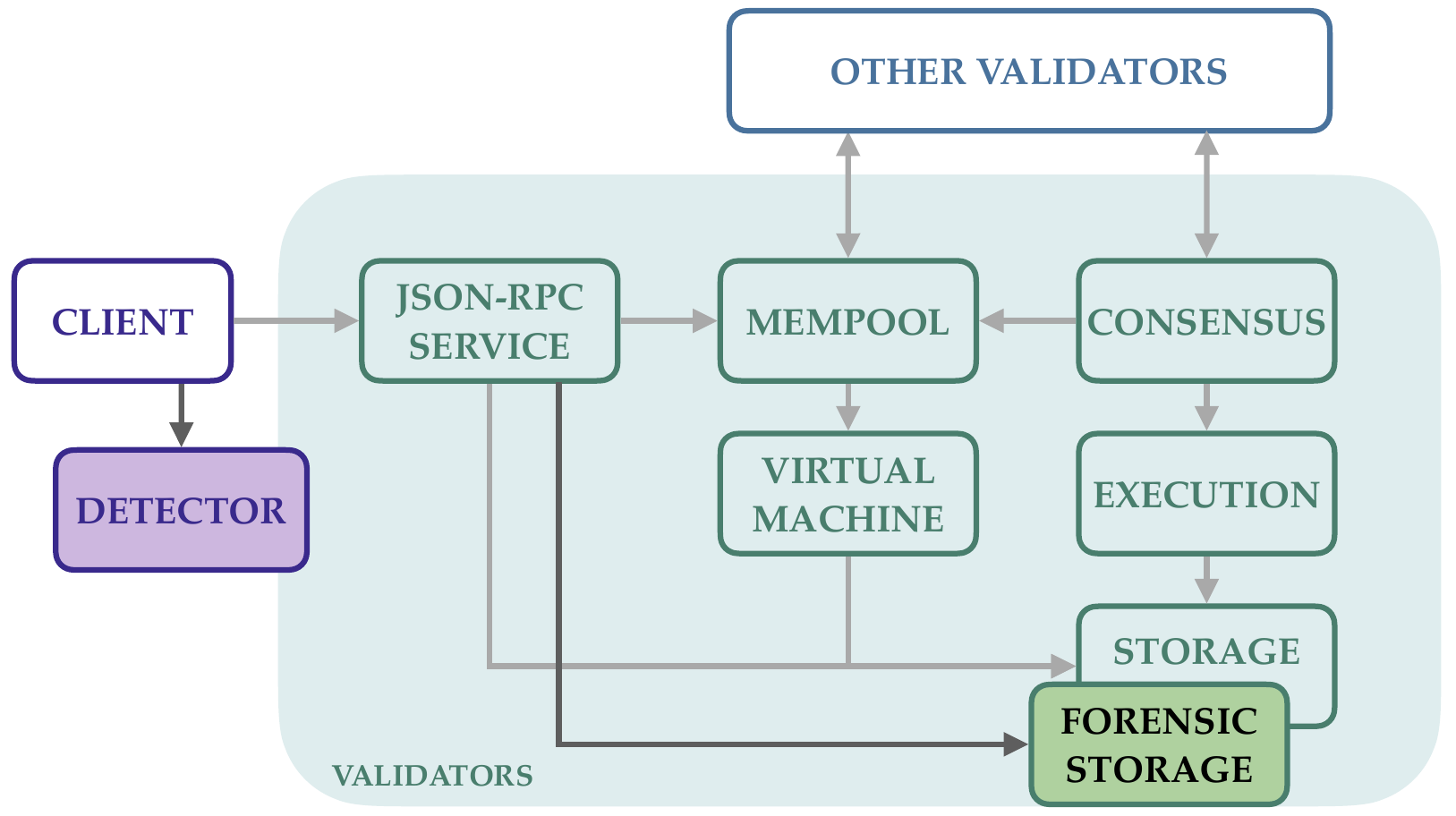}
    \caption{Forensic module integrated with {\sf Diem}.}
    \label{fig:forensic_store}
\end{figure}

\myparagraph{Testing using Twins~\cite{bano2020twins}.}
To test the correctness of forensic protocols, we build a testbed to simulate Byzantine attacks and construct different types of safety violations. Ideally, for modularity purposes, our testbed should not require us to modify the underlying consensus protocol to obtain Byzantine behavior. We leverage Twins~\cite{bano2020twins}, an approach to emulate Byzantine behaviors by running two instances of a node (i.e. replica) with the same identity. Consider a simple example setting with four nodes (denoted by $node0\sim 3$), where $node0$ and $node1$ are Byzantine so they have twins called $twin0$ and $twin1$. The network is split into two partitions, the first partition $P_1$ includes nodes $\{node0, node1, node2\}$ and the second partition $P_2$ includes nodes $\{twin0, twin1, node3\}$. Nodes in one partition can only receive the messages sent from the same partition. The double voting attack can be simulated when Byzantine leader proposes different valid blocks in the same view, and within each partition, all nodes will vote for the proposed block. The network partition is used to drop all messages sent from a set of nodes. However, it can only help construct the safety violation within the view. To construct more complicated attacks, we further improve the framework and introduce another operation called ``detailed drop'', which drops selected messages with specific types.

\begin{figure*}
    \centering
    \includegraphics[width=0.8\linewidth]{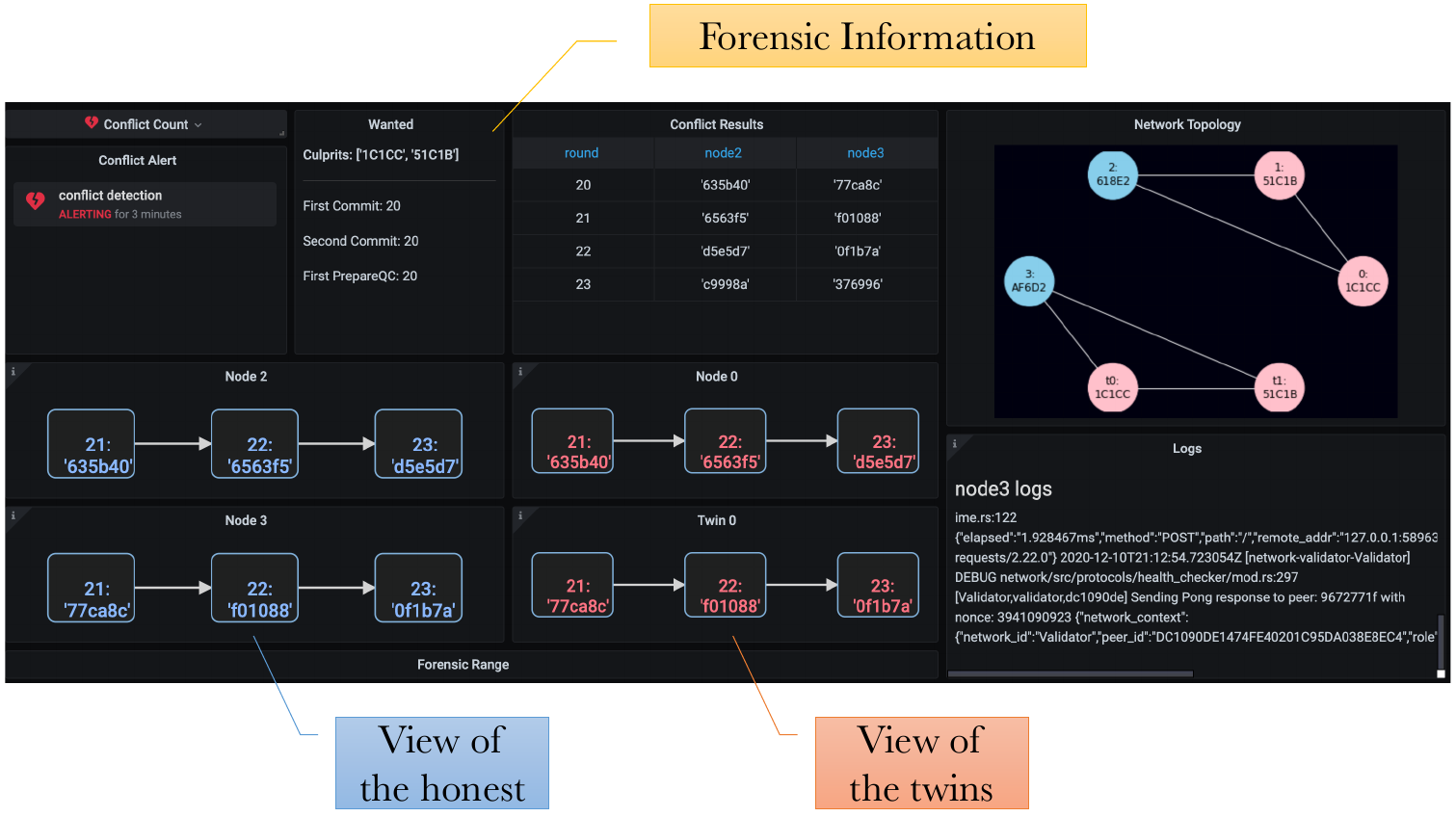}
    \caption{Forensic module dashboard.  }
    \label{fig:demo}
\end{figure*}
\myparagraph{Visualization.} The \detector accepts the registration of different views to get notified once the data is updated. We built  a dashboard to display the  information received by the detector and the analysis results output by the forensic protocol.  Figure~\ref{fig:demo} shows a snapshot of the dashboard   which displays  information about the network topology, hashes of latest blocks received at different validators, conflicting blocks, detected culprit keys and raw logs. Interactions with end-users,  including {\sf Diem} core-devs, has guided our design of the dashboard.

\section{Conclusion}
\label{sec:conclusion}
In this paper, we have embarked on a systematic study of the forensic properties of BFT protocols, focusing on 4 canonical examples: PBFT (classical), Hotstuff and VABA   (state-of-the-art protocols on partially synchronous and asynchronous network settings) and Algorand (popular protocol that is adaptable to proof of stake blockchains). Our results show that minor variations in the BFT protocols can have outsized impact on their forensic support. 

We exactly characterize  the forensic support of each protocol, parameterized by the triplet $(m,k,d)$.
The forensic support characterizations are remarkably similar across the protocols: if any non-trivial support is possible (i.e., at least one culpable replica can be implicated; $d> 0$), then the largest possible forensic support, $(2t, 1, t+1)$, is also possible; the one exception to this result is the Hotstuff-hash variant. Although the proof of forensic support is conducted for each protocol and its variant individually, we observe common trends: 
\begin{itemize}
    \item For each of the protocols with strong forensic support, as a part of the protocol execution, there exist witnesses who hold signed messages from Byzantine parties indicating that they have not followed some rule in the protocol. 
    \item On the other hand, for protocols with no forensic support, the Byzantine parties are able to break safety without  leaving any evidence, although the mechanism to achieve this is different for each of PBFT-MAC, Algorand, and HotStuff-null. With PBFT-MAC, Byzantine parties are able to construct arbitrary transcripts due to the absence of signatures. Hence, message transcripts cannot be used as evidence. With Algorand, they can utilize a rule which relies on the absence of messages (under synchrony) to set an incorrect protocol state without leaving a trail. With HotStuff-null, due to the lack of links between messages across views, Byzantine parties can present fake message transcripts and thus, pretend to be honest. 
\end{itemize}
Conceptually, the  burning question is whether these common ingredients can be stitched together to form an overarching theory of   forensic support  for abstract families of secure BFT protocols: First, from an impossibility standpoint, is there a relationship between the need to use synchrony or the absence of PKI in a protocol  and absence of  forensic support? Second, for the positive results, can one  argue strong forensic support for an ``information-complete'' variant of any BFT protocol?  This is an active area of research. 

From a practical stand point, forensic analysis for existing blockchain protocols is of great interest; our forensic protocol for LibraBFT and its reference implementation has made strong inroads towards practical deployment. However, one shortcoming of the approach  in this paper is that forensic analysis is conducted   only upon {\em fatal} safety breaches. It is of great interest to conduct forensics with other forms of attacks: liveness attacks, censorship, small number of misbehaving replicas that impact performance.  We note that liveness attacks do not  afford, at the outset, undeniable complicity of malicious replicas and an important research direction is in formalizing weaker notions of culpability  proofs (perhaps including assumptions involving external forms of trust). 
This is an active area of research. 
\section{Acknowledgement}
This research was partly supported by US Army Research Office Grant W911NF-18-1-0332, National Science Foundation CCF-1705007 and the XDC network. Kartik Nayak was supported in part by Novi and VMware gift research grant.

We thank the {\sf Diem} team for  implementation advice and Jovan Komatovic for helpful discussions about a forensic attack on HotStuff.
\bibliographystyle{ACM-Reference-Format}
\bibliography{acmart}
\appendix
\section{VABA has Strong Forensic Support}
\label{sec:vaba}
Validated Asynchronous Byzantine Agreement (VABA) is a  state-of-the-art protocol~\cite{abraham2019asymptotically} in the asynchronous setting with asymptotically optimal $O(n^2)$ communication complexity and expected $O(1)$ latency for $n \geq 3t+1$. 

\begin{algorithm}[b]
\begin{algorithmic}[1]
\As{a replica running VABA}
    \ForInitialize{$e\ge 1$}
        \For{$j\in[n]$}
        \State $\pflock[e][j]\gets \{\}$
        \State $\coin[e]\gets \{\}$
        \EndFor
    \EndForInitialize
    \Upon{$\lr{i,\nv,e,v,L}$ in view $e$ in \pp instance $i$}
        \State $(e', v', \sigma, \eqc)\gets L $\algorithmiccomment{Note that $L$ has selectors $e,v,\sigma,\eqc$} 
        \State $\pflock[e'][i] \gets  \pflock[e'][i]\cup \{(v', \sigma, \eqc)\}$\label{alg:vaba-6}
    \EndUpon
    \Upon{$\lr{i,\prepare,e,v,\sigma,\eqc}$ in view $e$ in \pp instance $i$}
        \State $\pflock[e][i]\gets \pflock[e][i]\cup \{(v, \sigma,\eqc)\}$\label{alg:vaba-6-2}
    \EndUpon
    \If{$Leader[e]$ is elected}
    \State discard $\pflock[e][j]$ for $j\neq Leader[e]$\label{alg:vaba-12}
    \State $\coin[e]\gets$ inputs to $\threscoin$ for electing $Leader[e]$\label{alg:vaba-13}
    \EndIf
    \Upon{$\lr{\vc,e,\qctwo,$ $\qcprecom,\qcthree}$ in view $e$}
        \State $(e', v', \sigma, \eqc)\gets\qctwo$\algorithmiccomment{Note that \qctwo has selectors $e,v,\sigma,\eqc$} 
        \State $\pflock[e'][Leader[e']]\gets \pflock[e'][Leader[e']]\cup \{(v', \sigma,\eqc)\}$\label{alg:vaba-6-3}
    \EndUpon
    \Upon{$\lr{\requestproofleader, e, e'}$ from a client}
        \For{all $e\le e^*\le e'$}
            \State send $\lr{\coinmsg, e^*, Leader[e^*], \coin[e^*]}$ to client
        \EndFor
    \EndUpon
    \Upon{$\lr{\conflictacross, e, v, \sigma, e'}$ with a collection of $LeaderMsg$ from a client}
        \For{all $e<e^\#\le e'$}
            \If{$Leader[e^\#]$ is not elected yet}
            \State $\lr{\coinmsg, e^\#, leader, \coin}\gets LeaderMsg$ of view $e^\#$
            \State check $leader$ is the leader generated by $\coin$ in view $e^\#$ (otherwise don't reply to the client)
            \State $Leader[e^\#]\gets leader$
            \EndIf
            \For{$qc\in \pflock[e^\#][Leader[e^\#]]$}
            \If{$(qc.v\neq v)\land (qc.\eqc\le e)$}
                \State send $\lr{\proofacross, e^\#, Leader[e^\#], qc}$ to the client
            \EndIf
            \EndFor
        \EndFor
    \EndUpon
\EndAs{}
\algstore{vaba}
\end{algorithmic}
\caption{Forensic protocol for VABA}
\label{alg:fa-vaba}
\end{algorithm}
\addtocounter{algorithm}{-1}
\begin{algorithm}[htbp]
\begin{algorithmic}[1]
\algrestore{vaba}
\As{a client}
    \Upon{two conflicting \reply messages}
        \State $e \gets$ the view number of \reply  from lower view
        \State $e'\gets$ the view number of \reply  from higher view
        \ForInitialize{all $e\le e^*\le e'$}
            \State $Leader[e^*]\gets \{\}$
            \State $LeaderMsg[e^*]\gets \{\}$
        \EndForInitialize
        \State send $\lr{\requestproofleader, e, e'}$  to the replica of \reply message from higher view\label{alg:vaba-request-proof-leader-1}
        \For{all $e\le e^*\le e'$}
        \State wait for $\lr{\coinmsg, e^*, leader, \coin}$ s.t. $leader$ is the leader generated by $\coin$ in view $e^*$ (otherwise the \reply message is not considered valid)
        \State $Leader[e^*]\gets leader$
        \State $LeaderMsg[e^*]\gets \lr{\coinmsg, e^*, leader, \coin}$
        \EndFor\label{alg:vaba-request-proof-leader-2}
        \If{the two \reply messages are from different views}\label{alg:vaba-28}
            \State $\lr{i,\reply,e,v,\sigma}\gets$ the message from lower view
            \State $\lr{i',\reply,e',v',\sigma'}\gets$ the message from higher view
            \State check $i=Leader[e]$ and $i'=Leader[e']$ (otherwise the \reply message is not considered valid)\label{alg:vaba-verify-leader-1}
            \State broadcast $\lr{\conflictacross, e, v, \sigma, e'}$ with $LeaderMsg[e^*]$ for all $e<e^*\le e'$ 
            \State wait for  $\lr{\proofacross, e^\#,leader, qc}$  s.t.
            \begin{enumerate}[leftmargin=1in]
                \item[(1)] $e<e^\#\le e'$, and
                \item[(2)] $(qc.v\neq v)\land (qc.\eqc\le e)$, and
                \item[(3)] $leader=Leader[e^\#]$\label{alg:vaba-verify-leader-2}
            \end{enumerate}
            \State \textbf{output} $qc.\sigma \cap \sigma$
        \Else
            \State $\lr{i,\reply,e,v,\sigma}\gets$ first \reply message
            \State $\lr{i',\reply,e,v',\sigma'}\gets$ second \reply message
            \State check $i=i'=Leader[e]$ (otherwise the \reply message is not considered valid)\label{alg:vaba-verify-leader-3}
            \State \textbf{output} $\sigma \cap \sigma'$
        \EndIf\label{alg:vaba-end}
    \EndUpon
\EndAs{}
\end{algorithmic}
\caption{Forensic protocol for VABA}
\label{alg:fa-vaba-second}
\end{algorithm}

\subsection{Overview}
At a high-level, the VABA protocol adapts HotStuff to the asynchronous setting. There are three phases in the protocol:
\begin{itemize}[leftmargin=*, topsep=0pt]
	\item[-] \textbf{Proposal promotion.} In this stage, each of the $n$ replicas run $n$ parallel HotStuff-like instances, where replica $i$ acts as the leader within instance $i$.
    \item[-] \textbf{Leader election.} 
		After finishing the previous stage, replicas run a leader election protocol using a \threscoin primitive \cite{cachin2005random} to randomly elect the leader of this view, denoted as $Leader[e]$ where $e$ is a view number. At the end of the view, replicas adopt the ``progress'' from $Leader[e]$’s proposal promotion instance, and discard values from other instances.
	\item[-] \textbf{View change.} Replicas broadcast quorum certificates from $Leader[e]$’s proposal promotion instance and update local variables and/or output value accordingly.
\end{itemize}

Within a proposal promotion stage, the guarantees provided are the same as that of HotStuff, and hence we do not repeat it here. The leader election phase elects a unique leader at random -- this stage guarantees (i) with $\geq 2/3$ probability, an honest leader is elected, and (ii) an adaptive adversary cannot stall progress (since a leader is elected in hindsight). Finally, in the view-change phase, every replica broadcasts the elected leader's quorum certificates to all replicas.

\subsection{Forensic Support for VABA}
The key difference between forensic support for HotStuff and VABA is the presence of the leader election stage -- every replica/client needs to know \emph{which} replica was elected as the leader in each view. Importantly, the \threscoin  primitive ensures that there is a unique leader elected for each view. 
Thus, the forensic analysis boils down to performing an analysis similar to the HotStuff protocol, except that the leader of a view is described by the leader election phase.
 
We present the full forensic protocol in Algorithm~\ref{alg:fa-vaba} for completeness. We make the following changes to VABA:
\begin{itemize}
    \item \textbf{Storing information for forensics.} Each replica maintains a list of \emph{ledgers} for all instances, containing all received \qctwo from \prepare messages, \nv messages, and \vc messages (lines~\ref{alg:vaba-6},\ref{alg:vaba-6-2},\ref{alg:vaba-6-3}). When the leader of a view is elected, a replica keeps the \pflock from the leader's instance and discards others (line~\ref{alg:vaba-12}). A replica also stores the random coins from the leader election phase for client verification (line~\ref{alg:vaba-13}).
    \item \textbf{Bringing proposal promotion closer to \hsa.} There are minor differences in the proposal promotion phase of VABA~\cite{abraham2019asymptotically} to the description in HotStuff (\S\ref{sec:basic-hotstuff}). We make this phase similar to that in the description of our HotStuff protocol with forensic support. In particular: (i) the $LOCK$ variable stores both the view number and the value (denoted by $LOCK.e$ and $LOCK.v$), (ii) the voting rule in a proposal promotion phase is: vote if $KEY$ has view and value equal to $LOCK$, except when $KEY$'s view is strictly higher than $LOCK.e$, (iii) assume a replica's own \vc message arrives first so that others' \vc messages do not overwrite local variables $KEY$ and $LOCK$, and (iv) add $\eqc$ into \vone. 
\end{itemize}

A client first verifies leader election (lines~\ref{alg:vaba-request-proof-leader-1}-\ref{alg:vaba-request-proof-leader-2}). Then, it follows steps similar to the HotStuff forensic protocol (lines~\ref{alg:vaba-28}-\ref{alg:vaba-end}) except that there are added checks pertaining to leader elections (lines~\ref{alg:vaba-verify-leader-1},\ref{alg:vaba-verify-leader-2},\ref{alg:vaba-verify-leader-3}).

We prove the forensic support in Theorem~\ref{thm:vaba}.
\begin{theorem}
\label{thm:vaba}
For $n=3t+1$, when $f > t$, if two honest replicas output conflicting values, VABA protocol provides $(2t,\ 1,\ t+1)$-forensic support. Further $(2t+1,\ n-f,\ d)$-forensic support is impossible with $d>0$.
\end{theorem}
\begin{proof} We prove the forward part of the theorem below. The proof of converse (impossibility) is the same as \S\ref{sec:proofoftheorem4.1}.

The leader of a view is determined by the threshold coin-tossing primitive $\threscoin$ and Byzantine replicas cannot  forge the result of a leader election by the robustness property of the threshold coin. 
Suppose two conflicting outputs happen in view $e,e'$ with $e\le e'$. The replica who outputs in view $e'$ has access to the proof of leader election of all views $\le e'$. Therefore, a client can verify the leader election when it receives messages from this replica. Even if other replicas have not received messages corresponding to the elections in views $\le e'$, the client can send the proof of leader to them. The remaining forensic support proof follows from Theorem~\ref{thm:hotstuff_a} in a straightforward manner, where any witness will receive the proof of leader from the client (if leader is not elected) and send the proof of culprits to the client.
\end{proof}

\myparagraph{Communication complexity.} The client needs to first receive all leader election results from view $e$ to view $e'$, and each result is of size $|coin|$ (the size of the coin in the \threscoin primitive). Then, the client shares leader election results with all replicas. This step incurs receiving message complexity $O(l|coin|)$ where $l=e'-e$. Next, the client needs to receive one message from $k=1$ replica and the message size is $(|v|+|sig|)$. Therefore the complexity for the client receiving messages is $O(|v|+|sig|+l|coin|)$. However, the procedure of sharing leader election is irrelevant to forensic support, and we could assign it to replicas. (This procedure is included in the forensic protocol because we do not want to change the consensus protocol itself.) In that case, the client needs to receive just one leader election result, so the receiving message complexity is $O(|v|+|sig|+|coin|)$. 
\section{Impossibility of Forensic Support for $n=2t+1$}\label{sec:lb}

A validated Byzantine agreement protocol allows replicas to obtain agreement, validity, and termination so far as the actual number of faults $f\leq t$ where $t$ is a Byzantine threshold set by the consensus protocol. A protocol that also provides forensic support with parameters $m$ and $d$ allows the detection of $d$ Byzantine replicas when $\leq m$ out of $n$ replicas are Byzantine faulty. In particular, in \S\ref{sec:pbft} and \S\ref{sec:basic-hotstuff}, we observed that when $t=\lfloor n/3\rfloor,\ m = 2t$, and $k=1$, we can obtain $(2t,1,d)$-forensic support for $d = t+1$. This section presents the limits on the number of Byzantine replicas detected ($d$), given the total number of Byzantine faulty replicas available in the system ($m$). In particular, we show that if the total number of Byzantine faults are too high, in case of a disagreement, the number of corrupt (Byzantine) replicas that can be deemed undeniably culpable will be too few.

\myparagraph{Intuition.} To gain intuition, let us consider a specific setting with $n = 2t+1$, $m = n-t = t+1$, and $d > 1$. Thus, such a protocol provides us with agreement, validity, and termination if the Byzantine replicas are in a minority. If they are in the majority, the protocol transcript provides undeniable guilt of more than one Byzantine fault. We show that such a protocol does not exist. Why? Suppose we split the replicas into three groups $P$, $Q$, and $R$ of sizes $t$, $t$, and $1$ respectively. First, observe that any protocol cannot expect Byzantine replicas to participate in satisfying agreement, validity, and termination. Hence, if the replicas in $Q$ are Byzantine, replicas in $P \cup R$ may not receive any messages from $Q$. However, if, in addition, the replica $R$ is also corrupt, then $R \cup Q$ can separately simulate another world where $P$ are Byzantine and not sending messages, and $Q \cup R$ output a different value. Even if an external client obtains a transcript of the entire protocol execution (i.e., transcripts of $k=n-f$ honest replicas and $f$ Byzantine replicas), the only replica that is undeniably culpable is $R$ since it participated in both worlds. For all other replicas, neither $ P $ nor $ Q $ have sufficient information to prove the other set's culpability. Thus, an external client will not be able to detect more than one Byzantine fault correctly.
Our lower bound generalizes this intuition to hold for $n > 2t$, $m = n-t$, $k = n-f$, and $d > n-2t$.

\begin{theorem} For any validated Byzantine agreement protocol with $t<n/2$, when $f > t$, if two honest replicas output conflicting values, $(n-t,\ n-f,\ d)$-forensic support is impossible with $d>n-2t$.
\end{theorem}

\begin{proof}
Suppose there exists a protocol that achieves agreement, validity, termination, and forensic support with parameters $n,\ t < n/2,\ m = n-t,\ k = n-f$ and $d > n-2t$. Through a sequence of worlds, and through an indistinguishability argument we will show the existence of a world where a client incorrectly holds at least one honest replica as culpable.
Consider the replicas to be split into three groups $P,\ Q$, and $R$ with $t$, $t$, and $n-2t$ replicas respectively. We consider the following sequence of worlds:

\myparagraph{World 1.} [$t$ Byzantine faults, satisfying agreement, validity, and termination]

\noindent\underline{Setup.} Replicas in $P$ and $R$ are honest while replicas in $Q$ have crashed. $P$ and $R$ start with a single externally valid input $v_1$. All messages between honest replicas arrive instantaneously.

\noindent\underline{Output.} Since there are $|Q| = t$ faults, due to agreement, validity and termination properties, replicas in $P$ and $R$ output $v_1$. Suppose replicas in $P$ and $R$ together produce a transcript $T_1$ of all the messages they have received. 

\myparagraph{World 2.} [$t$ Byzantine faults, satisfying agreement, validity, and termination]

\noindent\underline{Setup.} Replicas in $Q$ and $R$ are honest while replicas in $P$ have crashed. $Q$ and $R$ start with an externally valid input $v_2$. All messages between honest replicas arrive instantaneously.

\noindent\underline{Output.} Since $t$ replicas are Byzantine faulty, due to agreement, validity and termination properties, replicas in $Q$ and $R$ output $v_2$. Suppose replicas in $Q$ and $R$ together produce a transcript $T_2$ of all the messages they have received. 

\myparagraph{World 3.} [$n-t$ Byzantine faults satisfying validity, termination, forensic support]

\noindent\underline{Setup.} Replicas in $P$ are honest while replicas in $Q$ and $R$ are Byzantine. Replicas in $P$ start with input $v_1$. Replicas in $Q$ and $R$ have access to both inputs $v_1$ and $v_2$. $Q$ behaves as if it starts with input $v_2$ whereas $R$ will use both inputs $v_1$ and $v_2$.  Replicas in $Q$ and $R$ behave with $P$ exactly like in World~1. In particular, replicas in $Q$ do not send any message to any replica in $P$. Replicas in $R$ perform a split-brain attack where one brain interacts with $P$ as if the input is $v_1$ and it is not receiving any message from $Q$. 
Also, separately, replicas in $Q$ and the other brain of $R$ start with input $v_2$ and communicate with each other exactly like in World~2. They ignore messages arriving from $P$. 

\noindent\underline{Output.} For replicas in $P$, this world is indistinguishable from that of World~1. Hence, they output $v_1$. Replicas in $P$ and the first brain of $R$ output transcript $T_1$ corresponding to the output. Replicas in $Q$ and the other brain of $R$ behave exactly like in World~2. Hence, they can output transcript $T_2$. Since the protocol provides $(n-t,n-f, d)$-forensic support for $d > n-2t$, the transcript  of messages should hold $d > n-2t$ Byzantine replicas undeniably corrupt. Suppose the client can find the culpability of $> n-2t$ replicas from $Q \cup R$, i.e., $\geq 1$ replica from $Q$.

\myparagraph{World 4.} [$n-t$ Byzantine faults satisfying validity, termination, forensic support]

\noindent\underline{Setup.} Replicas in $Q$ are honest while replicas in $P$ and $R$ are Byzantine. Replicas in $Q$ start with input $v_2$. Replicas in $P$ and $R$ have access to both inputs $v_1$ and $v_2$. $P$ behaves as if it starts with input $v_1$ whereas replicas in $R$ use both $v_1$ and $v_2$. Replicas in $P$ and $R$ behave with $Q$ exactly like in World~2. In particular, replicas in $P$ do not send any message to any replica in $Q$. Replicas in $R$ perform a split-brain attack where one brain interacts with $Q$ as if the input is $v_2$ and it is not receiving any message from $P$. 
Also, separately, replicas in $P$ and the other brain of $R$ start with input $v_1$ and communicate with each other exactly like in World~1. They ignore messages arriving from $Q$.

\noindent\underline{Output.} For replicas in $Q$, this world is indistinguishable from that of World~2. Hence, they output $v_2$. Replicas in $Q$ and the first brain of $R$ output transcript $T_2$ corresponding to the output. Replicas in $P$ and the other brain of $R$ behave exactly like in World~1. Hence, they can output transcript $T_1$. 

Observe that the transcript and outputs produced by replicas in $P$, $Q$, and $R$ are exactly the same as in World~3. Hence, the client will hold $> n-2t$ replicas from $Q \cup R$, i.e., $\geq 1$ replica from $Q$ as culpable. However, all replicas in $Q$ are honest in this world. This is a contradiction.
\end{proof}
\section{Proof of Theorems}\label{sec:proof-diem}

\subsection{Proof of Theorem~\ref{thm:pbft}}
\label{sec:proofoftheorem4.1}

\begin{proof}

Suppose there are $f = 2t+1$ Byzantine replicas, and let there be three replica partitions $P, Q, R$, $|P|=|Q|=t$, $|R|=t+1$. To prove the result, suppose the protocol has forensic support for $d>0$, we construct two worlds where a different set of replicas are Byzantine in each world.

\myparagraph{World~1.} Let $R,Q$ be Byzantine replicas in this world. During the protocol, replicas in $Q$ behave like honest parties. Suppose in view $e, e'$ ($e < e'$), two honest replicas $p_1,p_2\in P$ output two conflicting values $v, v'$ after receiving two \qcthree. The \qcthree for $v$ contains the \commit messages from $P$ and $R$, the \qcthree for $v'$ contains the \commit messages from $R$ and $Q$. All the other messages never reach $P$. During the forensic protocol, replicas in $P$ send their transcripts to the client. Since the protocol has forensic support for $d>0$, using these transcripts (two \qcthree), the forensic protocol determines some subset of $R$ are culpable (since $Q$ behave like honest).

\myparagraph{World~2.} Let $P,Q$ and some replica $r\in R$ are Byzantine replicas and replicas in $r$ behave honestly. Again, in view $e, e'$ ($e < e'$), two replicas $p_1,p_2\in P$ output two conflicting values $v, v'$ after receiving two \qcthree. The \qcthree for $v$ contains the \commit messages from $P$ and $R$, the \qcthree for $v'$ contains the \commit messages from $R$ and $Q$. Replicas in $R$ unlock themselves due to receiving a higher \qctwo formed in $e^*$ ($e<e^*<e'$). During the forensic protocol, replicas in $P$ send the same transcripts as of World~1 to the client (only two \qcthree). Thus, the forensic protocol outputs some subset of $R$ as culpable replicas. However, this is incorrect since replicas in $R$ are honest ($r$ is indistinguishable with replicas in $R/\{r\}$). This complete the proof.
\end{proof}



\subsection{Proof of Theorem~\ref{thm:hotstuff-null}}
\label{proof:hotstuff-null}
\begin{figure}
    \centering
    \includegraphics[width=\columnwidth]{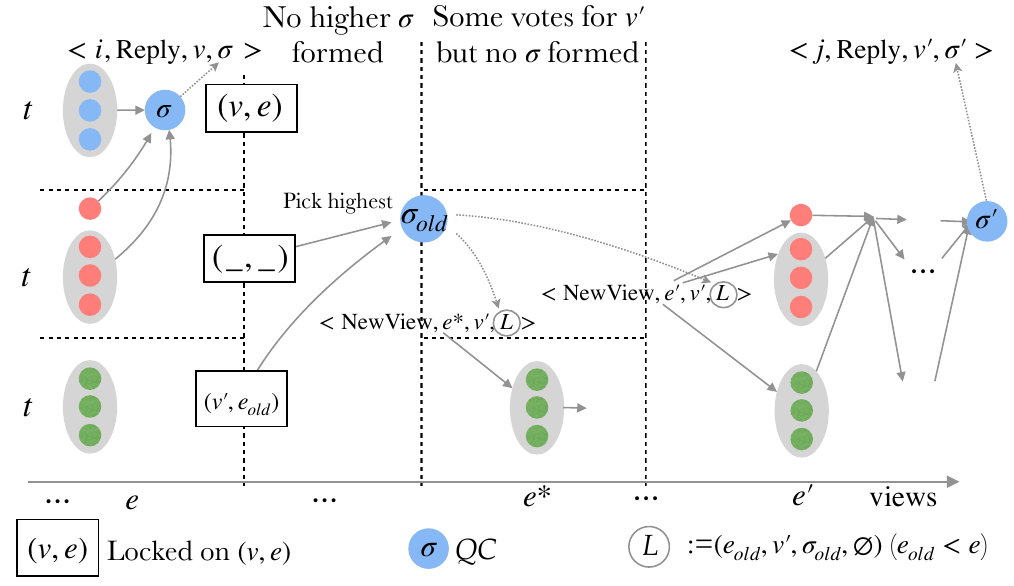}
    \caption{World~1 of Theorem~\ref{thm:hotstuff-null}. Replicas are represented as colored nodes. Replica partitions are $P$, $\{x\}$ (Byzantine), $R$ (Byzantine), and $Q$ from top to bottom.}
    \label{fig:hotstuff-null-1}
\end{figure}
\begin{figure}
    \centering
    \includegraphics[width=\columnwidth]{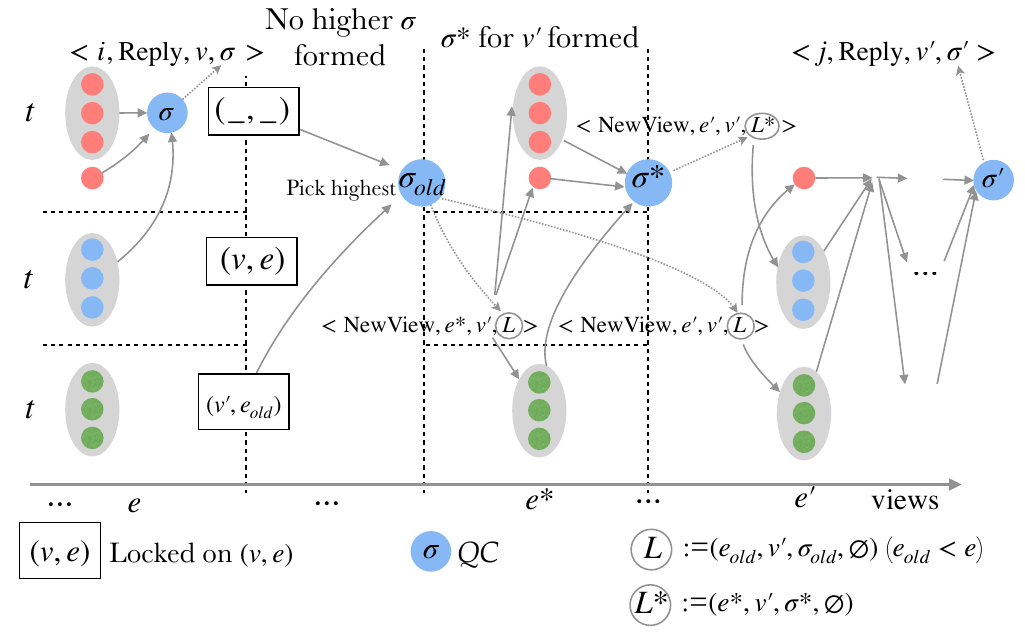}
    \caption{World~2 of Theorem~\ref{thm:hotstuff-null}. Replicas are represented as colored nodes. Replica partitions are $P$ (Byzantine), $\{x\}$ (Byzantine), $R$, and $Q$ from top to bottom.}
    \label{fig:hotstuff-null-2}
\end{figure}

\begin{proof}
Suppose the protocol provides forensic support to detect $d>1$ Byzantine replicas with irrefutable proof and the proof can be constructed from the transcripts of all honest replicas.  To prove this result, we construct two worlds where a different set of replicas are Byzantine in each world. We will fix the number of Byzantine replicas $f=t+1$, but the following argument works for any $f\ge t+1$.

Let there be four replica partitions $P,Q,R,\{x\}$. $|Q|=|P|=|R|=t$, and $x$ is an individual replica. In both worlds, the conflicting outputs are presented in view $e,e'$ ($e+1<e'$) respectively. Let $\qcthree_1$ be on value $v$ in view $e$, and signed by $P,R,x$. And let $\qcthree_2$ (and a \qcprecom) be on value $v'$ in view $e'$, and signed by $Q,R,x$. Suppose the leader of view $e,e^*,e'$ ($e<e^*<e'$) is replica $x$.

\myparagraph{World~1} is presented in Figure~\ref{fig:hotstuff-null-1}. 
Let $R$ and $x$ be Byzantine replicas in this world. In view $e^*$, the leader proposes value $v'$ and $Q$ sends \prepare on it, but a \qctwo is not formed. In view $e'$, the Byzantine parties, together with $Q$, sign \qctwo on $v'$. $e'$ is the first view where a \qctwo for $v'$ is formed.

During the forensic protocol, all honest replicas in $P$ and $Q$ send their transcripts. Byzantine $R$ and $x$ do not provide any information. Since the protocol has forensic support, the forensic protocol can output $d > 1$ replicas in $R$ and $x$ as culprits.

\myparagraph{World~2} is presented in Figure~\ref{fig:hotstuff-null-2}. 
Let $P$ and $x$ be Byzantine replicas in this world. Here, in view $e^* > e$, $P$ and $x$, together with $Q$ sign \qctwo on $v'$. In this world, $e^*$ is the first view where a \qctwo for $v'$ is formed. View $e' > e^*$ is similar to that of World~1 except that honest $R$ receives a \nv message with $\qctwo^*$ (rather than $\qctwo_{old}$).

During the forensic protocol, $Q$ sends their transcripts, which are identical to those in World~1. Byzantine $P$ can provide the same transcripts as those in World~1. Observe that the transcripts from $P$ and $Q$ presented to the forensic protocol are identical to those in World~1. Thus, the forensic protocol can also outputs $d>1$ replicas in $R$ and $x$ as culpable. In World~2, this is incorrect since replicas in $R$ are honest. 

Based on $2t$ transcripts, World~1 and World~2 are indistinguishable. To obtain an irrefutable proof of $d>1$ culprits, the client needs to collect more than $2t$ transcripts, more than the number of honest parties available. This completes the proof. 
\end{proof}

\paragraph{Remark.} The above proof can be easily modified to work with parameters $d > 0$ when $m = t+2$.

\subsection{Proof of Theorem~\ref{thm:algorand}}\label{sec:algorand-proof}
\begin{proof}
We construct two worlds where a different set of replicas are Byzantine in each world. Let replicas be split into three partitions $P$, $Q$, and $R$, and $|P|=(n-2\epsilon)/3,|Q|=|R|=(n+\epsilon)/3$ 
and $\epsilon>0$ is a small constant. Denote the numbers of replicas from $P,Q,R$ in a committee by $p,q,r$. Let $\kappa$ denote the expected committee size; $t_H=2\kappa/3$. With constant probability, we will have $p < \kappa/3$, $q > \kappa/3$ and $r > \kappa/3$ and $p+q < 2\kappa/3$ in steps 4 to 8. 

\myparagraph{World~1.} Replicas in $R$ are Byzantine in this world. We have $p+q < t_H$ and $q+r > t_H$. The Byzantine parties follow the protocol in \gradedconsensus. Thus, all replicas in step 4 hold the same tuple of $b=0$ and $v$ ($v\neq v_\perp$). Then, the following steps are executed.
\begin{itemize}
    \item[Step 4] Honest committee members that belong to $P$ and $Q$ broadcast their votes on $(b=0,v)$ whereas Byzantine committee members that belong to $R$ send votes to replicas in $P$ and not $Q$. 
    \item[Step 5] Replicas in $P$ satisfy Ending Condition 0, and output $b=0$ and the value $v$. Replicas in $Q$ do not receive votes from committee members in $R$, so they update $b=0$ and broadcast their votes on $(b=0, v)$. 
    Byzantine committee members that belong to $R$ pretend not to receive votes from committee members in $Q$, and also update $b=0$. And they send votes to replicas in $P$ and not $Q$.
    \item[Step 6] Replicas in $Q$ update $b=1$ since they receive $p+q<t_H$ votes.  Replicas in $R$ pretend not to receive votes from committee members in $Q$, and also update $b=1$. Committee members in $Q$ and $R$ broadcast their votes.
    \item[Steps 7-8] Committee members that belong to $Q$ and $R$ receive $q+r>t_H$ votes, so they update $b=1$ and broadcast their votes.
    \item[Step9] Replicas in $Q$ and $R$ satisfy Ending Condition 1, and output $b=1$ and $v_\perp$, a disagreement with replicas in $P$.
\end{itemize}
During the forensic protocol, replicas in $P$ send their transcripts and state that they have output $b=0$. $Q$ and $R$ send their transcripts claiming in steps 4 and 5 they do not hear from the other partition, and they state that output $b=1$. 

If this protocol has any forensic support, then it should be able to detect some replica in $R$ as Byzantine.

\myparagraph{World~2.} This world is identical to World~1 except (i) Replicas in $Q$ are Byzantine and replicas in $R$ are honest, and (ii) the Byzantine set $Q$ behaves exactly like set $R$ in World~1, i.e., replicas in $Q$ do not send any votes to $R$ in steps 4 and 5 and ignore their votes. 
During the forensic protocol, $P$ send their transcripts and state that they have output $b=0$. $Q$ and $R$ send their transcripts claiming in steps 4 and 5 they do not hear from the other partition, and they state that output $b=1$. 

From an external client's perspective, World~2 is indistinguishable from World~1. In World~2, the client should detect some replica in $R$ as Byzantine as in World~1, but all replicas in $R$ are honest. 
\end{proof}

\subsection{Proof of Theorem~\ref{thm:diem}}

\begin{proof}Suppose two conflicting blocks $b,\ b'$ are output in views $e$, $e'$ respectively.

\myparagraph{Case $e=e'$.}

\noindent\underline{Culpability.} The \qcthree of $b$ (the QC in $e+3$) and \qcthree of $b'$ intersect in $t+1$ replicas. These $t+1$ replicas should be Byzantine since the protocol requires a replica to vote for at most one value in a view. 

\noindent\underline{Witnesses.} Client can get the proof based on the two blocks in $e+3$, so additional witnesses are not necessary in this case.

\myparagraph{Case $e\neq e'$.}

\noindent\underline{Culpability.}
If $e \neq e'$, then WLOG, suppose $e < e'$. Since $b$ is output in view $e$, it must be the case that $2t+1$ replicas are locked on $(b, e)$ at the end of view $e$. Now consider the first view $e < e^* \leq e'$ in which a higher lock $(b'', e^*)$ is formed where $b'', b$ are not on the same chain (possibly $b''$ is on the chain of $b'$). Such a view must exist since $b'$ is output in view $e' > e$ and a lock will be formed in at least view $e'$. For a lock to be formed, a higher \qctwo must be formed too.

Consider the first view $e<e^\# \leq e'$ in which a \qctwo in chain of $b''$ is formed. The leader in $e^\#$ broadcasts the block containing a \hqc on $(b'', e'')$. Since this is the first time a higher \qctwo is formed and there is no \qctwo for chain of $b''$ formed between view $e$ and $e^\#$, we have $e''\le e$. The formation of the higher \qctwo indicates that $2t+1$ replicas received the block extending $b''$ with \hqc on $(b'', e'')$ and consider it a valid proposal, i.e., the view number $e''$ is larger than their locks because the block is on another chain. 

 Recall that the output block $b$ indicates $2t+1$ replicas are locked on $(b, e)$ at the end of view $e$. In this case, the $2t+1$ votes in \qctwo in view $e^\#$ intersect with the $2t+1$ votes in \qcthree in view $e$ at $t+1$ Byzantine replicas. These replicas should be Byzantine because they were locked on the block $b$ in view $e$ and vote for a conflicting block in a higher view $e^\#$ whose \hqc is from a view $e'' \leq e$. Thus, they have violated the voting rule.

\noindent\underline{Witnesses.} Client can get the proof by storing a \qctwo formed in $e^\#$ between $e$ and $e'$ in a different chain from $b$. The \qctwo is for the previous block in $e^\#$ whose \hqc is formed in a view $e'' < e$. For the replicas who have access to the \qctwo, they must have access to all blocks in the same blockchain. Thus only one witness is needed ($k=1)$ to provide the \qctwo and its previous block containing the \hqc on $(b'', e'')$, the \qctwo, \hqc and the first \qcthree act as the irrefutable proof.
\end{proof}

\section{Description of BFT Protocols}\label{sec:bft-protocols}

In the protocols in this paper, we assume that replicas and clients ignore messages with invalid signatures and messages containing external invalid values. When searching for an entity (e.g. lock or \qctwo) with the highest view, break ties by alphabetic order of the value. Notice that ties only occur when $f>t$ and Byzantine replicas deliberately construct conflicting quorum certificates in a view.

With $n=3t+1$, the descriptions of the PBFT protocol and HotStuff protocol are presented in Algorithm~\ref{alg:single-pbft} and \ref{alg:single-hotstuff-modified}.

\subsection{A Forensic Attack on HotStuff-view}
\label{sec:change}
Compared to HotStuff~\cite[Algorithm 2]{yin2019hotstuff}, Algorithm~\ref{alg:single-hotstuff-modified} highlights a slightly different voting rule in line~\ref{alg:change}. In addition to check whether $(LOCK.e < \hqc.e) \lor (LOCK.v = v)$ holds as in HotStuff~\cite[Algorithm 2]{yin2019hotstuff}, when the value in \nv is the same as the value in lock, our voting rule requires $LOCK.e = \hqc.e$.

We argue that the lack of this additional check on the view number will not affect the safety and liveness for HotStuff, but pose a threat for forensics. In the following, we exhibit a forensic attack on \hsa protocol with the original voting rule.
\begin{itemize}
    \item $e=i > 0:$ An honest replica $R$ receives a $\lr{\commit,i,v,\sigma}$ from the leader and updates its lock to be $(i, v, \sigma)$. $R$ sends $\lr{\commit,i,v}$ to leader, which is contained in a $\qcthree$ denoted as $qc_1$. $v$ is output in this view.
    \item $e=i+1:$ $R$ receives $\lr{\commit,i+1,v',\sigma'}$ and updates its lock to be $(i+1, v', \sigma')$. 
    \item $e = i+2:$ A leader broadcasts $\lr{\nv,i+2,v',\hqc}$, where \hqc is a QC from $i-1$. Replica $R$ receives the message and sends $\signed{\vone,i+2, v', i-1}$, because $LOCK.v = v'$ by checking the original voting rule. This message is contained into a $\qctwo$ denoted as $qc_2$. Further, $v'$ is output in this view.
\end{itemize}

In this execution, replica $R$ follows the protocol, however, it will be mistakenly blamed by Algorithm~\ref{alg:fa-hotstuff-a} if the client receives the $qc_1$ for $v$ and the $qc_2$ for $v'$. Since $qc_2.v \ne qc_1.v$ and $qc_2.e_{qc} = i-1 \le qc_1.e = i$ according to line~\ref{alg:hotstuffb-13}.

While the actual \qctwo whose intersection with $qc_1$ should be blamed is generated in $e=i+1$, it is possible that some honest replicas who have the same transcripts as $R$ will be improperly held culpable in this case. By adding the condition to check $LOCK.e = \hqc.e$, honest replicas will not vote for a \nv with stale \hqc, which prevents them from the attack described above.


\begin{algorithm*}[t]
\begin{algorithmic}[1]
    \State $LOCK\gets (0,v_\perp,\sigma_\perp)$ with selectors $e,v,\sigma$  \algorithmiccomment{$0, v_\perp,\sigma_\perp$: default view, value, and signature}
    \State $e\gets 1$
    \While{true}
    \LeftComment{\ppr and \vone Phase}
    \As{a leader}
        \State collect $\signed{\vc,e-1,\cdot}$ from $2t+1$ distinct replicas as status certificate $M$ \algorithmiccomment{Assume special \vc messages from view 0}
        \State $ v\gets$ the locked value with the highest view number in $M$
        \If{$v=v_\perp$}
        \State $v\gets v_i$
        \EndIf
        \State broadcast $\lr{\nv,e,v,M}$
    \EndAs{}
    \As{a replica} 
        \State wait for valid $\lr{\nv,e,v,M}$ from leader \algorithmiccomment{Use function \textsc{Valid($\lr{\nv,e,v,M}$)}}
        \State send $\signed{\vone,e,v}$ to leader
        
    \EndAs{}
    \LeftComment{\commit Phase}
    \As{a leader}
        \State collect $\signed{\vone,e,v}$ from $2t+1$ distinct replicas, denote the collection as $\Sigma$
        \State $\sigma\gets \thressign(\Sigma)$ 
        \State broadcast $\lr{\commit,e,v,\sigma}$
    \EndAs{}
    \As{a replica}
        \State wait for $\lr{\commit,e,v,\sigma}$ from leader \algorithmiccomment{$\qctwo$}
        \State $LOCK\gets (e,v,\sigma)$
        \State send $\signed{\commit,e,v}$ to leader
    \EndAs{}
    \LeftComment{\reply Phase}
    \As{a leader}
        \State collect $\signed{\commit,e,v}$ from $2t+1$ distinct replicas, denote the collection as $\Sigma$
        \State $\sigma\gets \thressign(\Sigma)$ 
        \State broadcast $\lr{\reply,e,v,\sigma}$
    \EndAs{}
\As{a replica}
        \State wait for $\lr{\reply,e,v,\sigma}$ from leader  \algorithmiccomment{$\qcthree$}
        \State output $v$ and send $\lr{\reply,e,v,\sigma}$ to the client
    \EndAs{}
    \State call procedure $\vc()$
    \EndWhile
    \State if a replica encounters timeout in any ``wait for'', call procedure $\vc()$
    \vspace{1em}
\Procedure{$\vc()$}{}
    \State broadcast $\signed{\blame,e}$
    \State collect $\signed{\blame,e}$ from $t+1$ distinct replicas, broadcast them
    \State quit this view
    \State send $\signed{\vc,e,LOCK}$ to the next leader
    \State enter the next view, $e\gets e+1$
\EndProcedure
\Function{\textsc{Valid}}{$\lr{\nv,e,v,M}$}
    \State $v^*\gets$ the locked value with the highest view number in $M$
    \If{$( v^*=v\lor v^*=v_\perp )\land$($M$ contains locks from $2t+1$ distinct replicas)}
    \State return $true$
    \Else
    \State return $false$
    \EndIf
\EndFunction
\caption{PBFT-PK protocol: replica's initial value $v_i$\ignore{\kartik{should move to appendix}\kartik{remove broadcast commit and add changes discussed in algorithm 1}\kartik{use e for view}}}
\label{alg:single-pbft}
\end{algorithmic}
\end{algorithm*}

\begin{algorithm*}[t]

\begin{algorithmic}[1]
\caption{General HotStuff protocol: replica's initial value $v_i$, protocol variant indicator $var\in$$ \{$`\hsa', `\hsb', `\hsc'$\}$}
\label{alg:single-hotstuff-modified}
    \State $\qctwo\gets (0,v_\perp,\sigma_\perp, \info_\perp)$ with selectors $e,v,\sigma,\info$ 
    \State $LOCK\gets (0,v_\perp)$ with selectors $e,v$  \algorithmiccomment{$0, v_\perp,\sigma_\perp, \info_\perp$: default view, value, signature, and info}
    \State $e\gets 1$
    \While{true}
    \LeftComment{\ppr and \prepare Phase}
    \As{a leader}
        \State collect $\lr{\vc,e-1,\cdot}$ from $2t+1$ distinct replicas as  $M$
        \algorithmiccomment{Assume special \vc messages from view 0}
        \State $\hqc\gets$ the highest QC in $M$
        \State $v\gets$ $\hqc.v$
        \If{$v=v_\perp$}
        \State $v\gets v_i$
        \EndIf
        \State broadcast $\lr{\nv,e,v,\hqc}$
    \EndAs{}
    \As{a replica}
        \State wait for  $\lr{\nv,e,v,\hqc}$ from leader s.t. $\hqc.v=v\lor \hqc.v=v_{\bot}$ \algorithmiccomment{Validate $v$}
            \If{$(LOCK.e < \hqc.e) \lor ( LOCK.v = v \land LOCK.e=\hqc.e)$}\label{alg:change}\Comment{Voting rule, see Appendix~\ref{sec:change}} 
                \State send $\signed{\vone,e,v, \textsc{Info}(var, \hqc)}$ to leader \Comment{Use function $\textsc{Info}(var, \hqc)$ }
            \EndIf
    \EndAs{}
    \LeftComment{\precommit Phase}
    \As{a leader}
        \State collect $\signed{\vone,e,v,\info}$ from $2t+1$ distinct replicas, denote the collection as $\Sigma$
        \State $\sigma\gets \thressign(\Sigma)$ 
        \State broadcast $\lr{\precommit,e,v,\sigma,\info}$
    \EndAs{}
    \As{a replica}
        \State wait for $\lr{\precommit,e,v,\sigma,\info}$ from leader \algorithmiccomment{$\qctwo$}
        \State $\qctwo\gets (e,v,\sigma, \info)$
        \State send $\signed{\precommit,e,v}$ to leader
    \EndAs{}
    \LeftComment{\commit Phase}
    \As{a leader}
        \State collect $\signed{\precommit,e,v}$ from $2t+1$ distinct replicas, denote the collection as $\Sigma$
        \State $\sigma\gets \thressign(\Sigma)$ 
        \State broadcast $\lr{\commit,e,v,\sigma}$
    \EndAs{}
\As{a replica}
        \State wait for $\lr{\commit,e,v,\sigma}$ from leader  \algorithmiccomment{$\qcprecom$}
        \State $LOCK\gets (e,v)$
        \State send $\signed{\commit,e,v}$ to leader
    \EndAs{}
    \LeftComment{\reply Phase}
    \As{a leader}
        \State collect $\signed{\commit,e,v}$ from $2t+1$ distinct replicas, denote the collection as $\Sigma$
        \State $\sigma\gets \thressign(\Sigma)$ 
        \State broadcast $\lr{\reply,e,v,\sigma}$
    \EndAs{}
    \As{a replica}
        \State wait for $\lr{\reply,e,v,\sigma}$ from leader \algorithmiccomment{$\qcthree$}
        \State output $v$ and send $\lr{\reply,e,v,\sigma}$ to the client
    \EndAs{}
    
    \State call procedure $\vc()$
    \EndWhile
    \State if a replica encounters timeout in any ``wait for'', call procedure $\vc()$
    \algstore{hotstuff} 
\end{algorithmic}
\end{algorithm*}

\begin{algorithm*}[t]
\begin{algorithmic}[1]
\algrestore{hotstuff} 
\Procedure{$\vc()$}{}
    \State broadcast $\signed{\blame,e}$
    \State collect $\signed{\blame,e}$ from $t+1$ distinct replicas, broadcast them
    \State quit this view
    \State send $\lr{\vc,e,\qctwo}$ to the next leader
    \State enter the next view, $e\gets e+1$
\EndProcedure
\Function{\textsc{Info}}{$var, \hqc$}\Comment{$var \in \{$`\hsa', `\hsb', `\hsc'$\}$}
    \If{$var = $`\hsa'}
        \State return $\hqc.e$
    \EndIf
    \If{$var = $`\hsb'}
        \State return $\text{Hash}(\hqc)$
    \EndIf
    \If{$var = $`\hsc'}
        \State return $\emptyset$
    \EndIf
\EndFunction

\end{algorithmic}
\end{algorithm*}
\end{document}